\pdfoutput=1
\documentclass[12pt,a4paper]{article}
\usepackage{jheplike,mathtools,todonotes,tikz}
\usepackage{amsmath,amsfonts,amsthm,dsfont}

\hypersetup{pdftitle=Banana,pdfcreator={},
            linkcolor=[rgb]{0.15,0.35,0.75},
            colorlinks=true,citecolor=[rgb]{0.675,0,0.2},
            urlcolor=[rgb]{0.15,0.35,0.65}}

\newcommand{\eqs}[1]{\begin{gather}\begin{aligned}#1\end{aligned}\end{gather}}
\newcommand{\ddl}{\mathfrak{D}^d\ell}

\newcommand{\cI}{\mathcal{I}}
\newcommand{\cJ}{\mathcal{J}}
\newcommand{\cR}{\mathcal{R}}
\newcommand{\cL}{\mathcal{L}}

\newcommand{\cM}{\mathcal{M}}

\newcommand{\ord}{\mathcal{O}}
\newcommand{\cQ}{\mathcal{Q}}
\newcommand{\cS}{\mathcal{S}}

\newcommand{\eps}{\epsilon}
\newcommand{\ivec}{\begin{pmatrix}\cI_1(\eps;x)\\\cI_2(\eps;x)\\\cI_3(\eps;x)\end{pmatrix}}

\newcommand{\bZ}{\mathbb{Z}}
\newcommand{\bQ}{\mathbb{Q}}
\newcommand{\bC}{\mathbb{C}}

\newcommand{\bP}{\mathbb{P}}
\newcommand{\slz}{\mathrm{SL}_2(\bZ)}
\newcommand{\pslz}{\mathrm{PSL}_2(\bZ)}

\newcommand{\abcd}{\left(\begin{smallmatrix}a&b\\c&d\end{smallmatrix}\right)}

\newcommand{\fH}{\mathfrak{H}}

\makeatletter
\g@addto@macro\bfseries{\boldmath}
\makeatother

\newif\ifnote 
\notetrue

\allowdisplaybreaks

\def\beq{\begin{equation}}
\def\eeq{\end{equation}}
\def\bsp#1\esp{\begin{split}#1\end{split}}

\DeclareMathOperator{\EK}{K}

\newcommand{\ban}{\textsf{ban}}
\newcommand{\sun}{\textsf{sun}}

\newcommand{\tz}{\xi}
\newcommand{\mH}{\aleph}

\newcommand{\HP}{\mathfrak{H}}
\newcommand{\eHP}{\mathfrak{H}^*}

\theoremstyle{definition}
\newtheorem{defi}{Definition}
\newtheorem{thm}{Theorem}

\newtheorem{lemma}{Lemma}
\newtheorem{proposition}{Proposition}

\newtheorem{rmk}{Remark}

\usepackage{nicefrac}
\newcommand{\qu}{\nicefrac{1}{4}}
\newcommand{\tbt}{\mathsf{2\times2}}
\usepackage{url}

\title{Meromorphic modular forms and the three-loop equal-mass banana integral}

\author[a]{Johannes Broedel,} 
\author[b]{Claude Duhr,}
\author[c]{Nils Matthes,}

\affiliation[a]{Institute for Theoretical Physics, ETH Zurich, Wolfgang-Pauli-Str.~27, 8093 Zürich, Switzerland} 
\affiliation[b]{Bethe Center for Theoretical Physics, Universit\"at Bonn, D-53115, Germany} 
\affiliation[c]{Mathematical Institute, University of Oxford, Andrew Wiles Building, Radcliffe Observatory Quarter, Woodstock Road, Oxford OX2 6GG, United Kingdom}

\emailAdd{jbroedel@ethz.ch}
\emailAdd{cduhr@uni-bonn.de}
\emailAdd{nils.matthes@maths.ox.ac.uk}

\abstract{We consider a class of differential equations for multi-loop Feynman integrals which can be solved to all orders in dimensional regularisation in terms of iterated integrals of meromorphic modular forms. We show that the subgroup under which the modular forms transform can naturally be identified with the monodromy group of a certain second-order differential operator. We provide an explicit decomposition of the spaces of modular forms into a direct sum of total derivatives and a basis of modular forms that cannot be written as derivatives of other functions, thereby generalising a result by one of the authors form the full modular group to arbitrary finite-index subgroups of genus zero. Finally, we apply our results to the two- and three-loop equal-mass banana integrals, and we obtain in particular for the first time complete analytic results for the higher orders in dimensional regularisation for the three-loop case, which involves iterated integrals of meromorphic modular forms.}

\preprint{\begin{minipage}[t]{8cm}\begin{flushright}BONN-TH-2021-08\\
          \end{flushright}\end{minipage}}

\begin{document}
\maketitle\thispagestyle{empty}



\section{Introduction}

Feynman integrals are a cornerstone of perturbative computations in Quantum Field Theory, and so it is important to have a good knowledge of the mathematics underlying them, including efficient techniques for their computation and a solid understanding of the space of functions needed to express them. The simplest class of functions that arise in Feynman integral computations are multiple polylogarithms (MPLs)~\cite{Lappo:1927,GoncharovMixedTate,Goncharov:1998kja} (see also refs.~\cite{Remiddi:1999ew,Gehrmann:2000zt,Ablinger:2011te}). The success of MPLs in Feynman integral computations can to a large extent be traced back to the fact that their algebraic properties are well understood (see, e.g., ref.~\cite{Duhr:2014woa}), and there are several efficient public implementations for their numerical evaluation~\cite{Gehrmann:2001pz,Gehrmann:2001jv,Buehler:2011ev,Vollinga:2004sn,Frellesvig:2016ske,Ablinger:2018sat,Naterop:2019xaf}. Moreover, it is well known that Feynman integrals satisfy systems of coupled first-order differential equations~\cite{Kotikov:1990kg,Kotikov:1991hm,Kotikov:1991pm,Gehrmann:1999as}, and MPLs are closely connected to the concepts of pure functions~\cite{ArkaniHamed:2010gh} and canonical differential equations~\cite{Henn:2013pwa}. It is fair to say that, whenever one can find a system of canonical differential equations that can be solved in terms of MPLs, the problem can be considered solved. 

However, it was realised already early on that MPLs do not suffice to express solutions to higher-loop Feynman diagrams \cite{Sabry,Broadhurst:1987ei,Bauberger:1994by,Bauberger:1994hx,Laporta:2004rb,Kniehl:2005bc,Aglietti:2007as,Czakon:2008ii,Brown:2010bw,MullerStach:2011ru,CaronHuot:2012ab,Huang:2013kh,Brown:2013hda,Nandan:2013ip,Ablinger:2017bjx}, though no analytic results in terms of a well-defined class of functions was available. The situation changed less than a decade ago, when it was shown that the two-loop sunrise integral can be expressed in terms of so-called elliptic dilogarithms~\cite{Bloch:2013tra,Adams:2015ydq,Adams:2016xah,Adams:2014vja,Adams:2013nia,Adams:2013kgc,Adams:2015gva,Broedel:2017siw}. 
The elliptic dilogarithm is a special case of elliptic multiple polylogarithms~\cite{BeilinsonLevin,LevinRacinet,BrownLevin,Broedel:2017kkb}, which also play a prominent role in the context of string amplitudes at one-loop, cf.~e.g., refs.~\cite{Broedel:2014vla,Broedel:2015hia,Broedel:2017jdo}. 
Soon after, it was realised that in the equal-mass case the two-loop sunrise integral can also be expressed as iterated integrals of modular forms~\cite{Adams:2017ejb,Broedel:2018iwv}. This class of functions is also of interest in pure mathematics~\cite{ManinModular,Brown:mmv,Matthes:QuasiModular,Brown:mmv2,matthes2021iterated}, and it is understood how to manipulate and evaluate these integrals efficiently~\cite{Duhr:2019rrs,Walden:2020odh}. More generally, it was suggested that modularity is an important feature of Feynman integrals associated to families of elliptic curves~\cite{Weinzierl:2020fyx}.

Despite all this progress in understanding Feynman integrals that do not evaluate to MPLs, there are still many questions left unanswered, and no general and algorithmic solution to evaluate and manipulate them is known, contrary to the case of ordinary MPLs. For example, while the importance of iterated integrals of modular forms is by now well established, the reason for why modular forms appear in the first place, and if so of which type, is not completely settled in the literature, and there was even an argument in the literature as to which congruence subgroup to attach to the two-loop sunrise integral~\cite{Adams:2017ejb,Frellesvig:2021vdl}. Also the link between differential equations and the appearance of these functions is not completely satisfactory (though there are indications that the concepts of pure functions and canonical forms known from MPLs carry over to Feynman integrals associated to families of elliptic curves~\cite{Broedel:2018qkq,Adams:2018yfj,Bogner:2014mha}). Finally, and probably most importantly, holomorphic modular forms are not sufficient to cover even the simplest cases of Feynman integrals depending on one variable. Indeed, it is known that, while in general higher-loop analogues of the sunrise integral -- the so-called $l$-loop banana integrals -- are associated to families of Calabi-Yau $(l-1)$-folds~\cite{Bloch:2014qca,Bloch:2016izu,Bourjaily:2018ycu,Bourjaily:2018yfy,Klemm:2019dbm,Bonisch:2020qmm,Bonisch:2021yfw}, the three-loop equal-mass banana integral in $D=2$ dimensions can be expressed in terms of the same class of functions as the two-loop equal-mass sunrise integral~\cite{Bloch:2014qca,Bloch:2016izu,Broedel:2019kmn}. However, if higher terms in the $\eps$-expansion in dimensional regularisation are considered, new classes of iterated integrals are required, which cannot be expressed in terms of modular forms alone.

The goal of this paper is to describe the (arguably) simplest class of differential equations beyond MPLs for which the space of solutions can be explicitly described, to all orders in the $\eps$-expansion. The relevance of this class of differential equations for Feynman integrals stems from the fact that they cover in particular the two- and three-loop equal-mass banana integrals. Their solution space can be described in terms of iterated integrals of meromorphic modular forms, introduced and studied by one of us in the context of the full modular group $\slz$~\cite{matthes2021iterated}. For Feynman integrals, however, modular forms for the full modular group are insufficient. We extend the results of ref.~\cite{matthes2021iterated} to arbitrary finite-index  subgroups of genus zero, and we provide in particular a basis for the algebra of iterated integrals they span. Our construction also naturally provides an identification of the type of modular forms required, namely those associated to the monodromy group of the associated homogeneous differential operator. This explains in particular the origin and the type of iterated integrals of modular forms encountered in Feynman integral computations. As an application of our formalism, we provide for the first time complete analytic for results for all master integrals of the three-loop equal-mass banana integrals in dimensional regularisation, including the higher orders in the $\eps$-expansion, and we see the explicit appearance of iterated integrals of meromorphic modular forms.


The article is organized as follows: in section \ref{sec:DEQs} we review general material on Feynman integrals and the differential equations they satisfy, and we describe the class of differential operators that we consider. In section~\ref{sec:modular} we review modular and quasi-modular forms. Section~\ref{sec:mero_sec} presents the main results of this paper, and we consider iterated integrals of meromorphic modular forms and present the main theorems. In section~\ref{sec:sunban} we calculate the monodromy groups for the equal-mass two- and three-loop and banana integrals, while section~\ref{sec:bananameromorphic} is devoted to framing the higher-orders in $\eps$ results for the three-loop banana integrals in terms of iterated integrals  of meromorphic modular forms. We include several appendices. In appendix~\ref{app:mathy} we present a rigorous mathematical proof the main theorem from section~\ref{sec:mero_sec}, and in appendix~\ref{app:sunban} we collect formulas related to the sunrise and banana integrals

\vspace{1mm}\noindent


\section{Differential equations and modular parametrisations}
\label{sec:DEQs}

\subsection{Feynman integrals and differential equations}
\label{sec:FIs_and_deqs}
The goal of this paper is to study a certain class of Feynman integrals and to characterize the functions necessary for their evaluation. We work in dimensional regularisation in $d=d_0-2\eps$ dimensions, where $d_0$ is an even integer. The Feynman integrals to be considered depend on a single dimensionless variable $x$ or -- equivalently -- two dimensionful scales. 

It is well known that using integration-by-parts identities~\cite{Chetyrkin:1981qh,Tkachov:1981wb}, all Feynman integrals that share the same set of propagators raised to different integer powers can be expressed as linear combinations of a small set of so-called \emph{master integrals}. Those master integrals satisfy a system of first-order linear differential equations of the form~\cite{Kotikov:1990kg,Kotikov:1991hm,Kotikov:1991pm,Gehrmann:1999as,Henn:2013pwa}
\beq\label{eq:DEQ_generic}
\partial_x\cI(x,\eps) = A(x,\eps)\cI(x,\eps) + \mathcal{N}(x,\eps)\,,
\eeq
where $\cI(x,\eps) = (I(x,\eps),\partial_xI(x,\eps),\ldots,\partial_x^{r-1}I(x,\eps))^T$ is the vector of independent master integrals depending on the maximal set of propagators in the family. $\mathcal{N}(x,\eps)$ is an inhomogeneous term stemming from integrals with fewer propagators, which we assume to be known and expressible to all orders in the dimensional regulator $\eps$ as a linear combination with rational functions in $x$ as coefficients of multiple polylogarithms (MPLs), defined by:
\beq\label{eq:MPL_def}
G(a_1,\ldots,a_n;x) = \int_0^x\frac{dt}{t-a_1}\,G(a_2,\ldots,a_n;t)\,,
\eeq
where the $a_i$ are complex constants that are independent of $x$. 
The entries of $A(x,\eps)$ are rational functions in $x$ and $\eps$. 
The differential equation~\eqref{eq:DEQ_generic} is equivalent to the inhomogeneous differential equation:
\beq\label{eq:DEQ_generic_higher}
\cL_{x,\eps}^{(r)}I(x,\eps) = N(x,\eps)\,,
\eeq
where $\cL_{x,\eps}^{(r)}$ is a differential operator of degree $r$ whose coefficients are rational functions in $x$ and $\eps$. 

In order to solve the differential equation~\eqref{eq:DEQ_generic}, we first note that it is always possible to choose the basis of master integrals such that the matrix $A(x,\eps)$ is finite as $\eps\to0$ (see, e.g., refs.~\cite{Chetyrkin:2006dh,Lee:2019wwn}). In that case we can change the basis of master integrals according to 
\beq\label{eq:change_IJ}
\cI(x,\eps) = W_r(x)\cJ(x,\eps)\,,
\eeq
where $W_r(x)$ is the Wronskian matrix of the homogeneous part of eq.~\eqref{eq:DEQ_generic_higher} at $\eps=0$, 
\beq\label{eqn:DEgen}
\cL_x^{(r)}u(x) = 0\,,
\eeq
where $\cL_x^{(r)}= \cL_{x,\eps=0}^{(r)}$. Let us write 
\beq\label{eq:cL_with_coefficients}
\cL^{(r)}_x = \sum_{j=0}^r a_j(x)\partial_x^j\,,
\eeq
with $a_j(x)$ being rational functions and $a_r(x)=1$. If we denote the solution space of $\cL_x^{(r)}$ by 
\beq\label{eq:sol_space}
\textrm{Sol}(\cL_x^{(r)}) = \bigoplus_{s=1}^r\mathbb{C}\psi_s(x)\,,
\eeq
then the Wronskian is $W_r(x) = (\psi_s^{(p-1)}(x))_{1\le p,s\le r}$, where $\psi_s^{(p)}(x) := \partial_x^p\psi_s(x)$. The Wronskian is in fact the matrix for a basis of maximal cuts for $\cI(x,\eps)$~\cite{Primo:2016ebd,Frellesvig:2017aai,Harley:2017qut,Bosma:2017ens}. 

After the change of variables in eq.~\eqref{eq:change_IJ}, the differential equation for $\cJ(x,\eps)$ takes the form
\beq\bsp
\partial_x\cJ(x,\eps) &\,= W_r(x)^{-1}\big(A(x,\eps)-A(x,0)\big)W_r(x)\,\cJ(x,\eps) + W_r(x)^{-1}\mathcal{N}(x,\eps)\\
&\, =\eps \widetilde{A}(x,\eps)\,\cJ(x,\eps) + \widetilde{\mathcal{N}}(x,\eps)\,.
\esp\eeq
This solution to the above system can be written as a path-ordered exponential:
\beq\label{eq:path_ordered}
\cJ(x,\eps) = \mathbb{P}\,\exp\left[\eps\int_{x_0}^xdx'\,\widetilde{A}(x',\eps)\right]\cJ(x_0,\eps)\,.
\eeq
The path-ordered exponential can easily be expanded into a series in $\eps$, and the coefficients of this expansion involve iterated integrals over one-forms multiplied by polynomials in the entries of the Wronskian. We see that in our setting where the $\eps$-expansion of the differential operator $\cL_{x,\eps}^{(r)}$ and the inhomogeneity $\mathcal{N}(x,\eps)$ involve rational functions and MPLs only, the class of iterated integrals needed to express $\cJ(x,\eps)$ is determined by the solution space of $\cL_x^{(r)}$ in eq.~\eqref{eq:sol_space}. It is an interesting question when these iterated integrals can be expressed in terms of other classes of special functions studied in the literature. For example, in the case where $\widetilde{A}(x,\eps)$ is rational in $x$, these iterated integrals can be evaluated in terms of MPLs. In general, however, little is known about these iterated integrals. 

The main goal of this paper is to discuss a certain class of differential equations where the resulting iterated integrals can be completely classified and an explicit basis can be constructed algorithmically. Before we describe this class of differential equations, we need to review some general material on linear differential equations.

\subsection{Linear differential operators and their monodromy group}
\label{sec:frobenius_review}
A point $x_0$ is called a \textit{singular} point of eq.~\eqref{eqn:DEgen} if one of the coefficient functions $a_j(x)$ in eq.~\eqref{eq:cL_with_coefficients} has a pole at $x_0$. The point at infinity is called \emph{singular} if after a change of variables $x\to 1/y$ in eq.~\eqref{eqn:DEgen} the transformed equation has a pole at $y=0$ in one of the coefficient functions. Points which are not singular are called \textit{ordinary} points of the differential equation. A singular point $x_i$ is called \textit{regular}, if the coefficients $a_{r-j}$ have a pole of at most order $j$ at $x_i$. If all singular points of eq.~\eqref{eqn:DEgen} are regular, the equation is called \textit{Fuchsian}. In the following we only discuss Fuchsian differential equations, and we denote the (finite) set of singular points by $\Sigma:=\{x_0,\ldots,x_{q-1}\}\subset \mathbb{P}^1_{\mathbb{C}}$, and use the notation $X:=\mathbb{P}^1_{\mathbb{C}}\setminus\Sigma$.
The differential operators obtained from Feynman integrals are expected to be of Fuchsian type. 

For every point $y_0\in\mathbb{P}^1_{\mathbb{C}}$ of a Fuchsian differential operator, the \textit{Frobenius method} can be used to construct a series representation of $r$ independent local solutions in a neighbourhood of this point. The starting point is the \textit{indicial polynomial} of a point, which can be obtained as follows: The differential equation~\eqref{eqn:DEgen} is equivalent to $\widetilde{\cL}_x^{(r)}u(x)=0$, where $\widetilde{\cL}_x^{(r)}$  has the form
\beq
	\widetilde{\cL}^{(r)}_x = \sum_{j=0}^r \tilde{a}_j(x)\theta_x^j\,,\qquad \theta_x = x\partial_x\,,
\end{equation}
where the $\tilde{a}_j(x)$ are polynomials that are assumed not to have a common zero. Note that the singular points are precisely the zeroes of $\tilde{a}_r(x)$. The indicial polynomial of $\widetilde{\cL}^{(r)}_x$ at $y_0=0$ is then  
$P_0(s) = \sum_{j=0}^r \tilde{a}_j(0)s^j$. The roots $s_i$ of $P_0(s)$ are called the \emph{indicials} or \emph{local exponents} at $0$. The indicial polynomial $P_{y_0}(s)$ and the local exponents at another point $y_0$ can be obtained by changing variables to $y=x-y_0$ (or $y=1/x$ if $y_0=\infty$). The local exponents characterise the solution space locally close to the point $y_0$ in the form of convergent power series. More precisely, if $y_0\in X = \mathbb{P}^1_{\mathbb{C}}\setminus \Sigma$ is a regular point, then $P_{y_0}(s)$ has degree $r$, and so there are precisely $r$ local exponents $s_1,\ldots,s_r$ (counted with multiplicity). Correspondingly, there are $r$ linearly independent power series solutions $\phi_i(y_0;x)$, $i\in\{1,...,r\}$, to eq.~\eqref{eqn:DEgen} of the form
\beq
\label{eqn:powerseriessolutions}
\phi_i(y_0;x) =  \sum_{n\ge 0}c_{i,n}(x-y_0)^{s_i+n}\,, \qquad c_{i,0}=1 \,.
\eeq
Note that this representation holds for $y_0\neq\infty$; if $y_0=\infty$, the expansion parameter is $1/x$. 

If $y_0\in\Sigma$ is a singular point, the degree of the indicial polynomial is less than $r$, and so there are less than $r$ local exponents (even when counted with multiplicity) and thus less than $r$ local solutions of the form \eqref{eqn:powerseriessolutions}. 
Without loss of generality we assume $y_0=x_0\in\Sigma$. The missing solutions generically exhibit a logarithmic behaviour as one approaches $x_0$. In particular, in the case of a single local exponent $s_1$, there is a single power series solution, and a tower of $(r-1)$ logarithmic solutions (we only consider the case $x_0\neq\infty$)
\beq\bsp\label{eq:Frob_log_solution}
\phi_i(x_0;x) &\, = (x-x_0)^{s_1}\sum_{k=1}^{i} \frac{1}{(k-1)!}\log^{k-1}(x-x_0)\, \sigma_{k}(x_0;x)\,,
\esp\eeq
where the $\sigma_{k}(x_0;x)$ are holomorphic at $x=x_0$.
A singular point with such a hierarchical logarithmic structure of solutions with $s_1$ an integer is called a point of \emph{maximal unipotent monodromy} (MUM-point). 

The power series obtained from the Frobenius method have finite radius of convergence: the solutions $\phi(y_0;x):=(\phi_i(y_0;x))_{1\le i\le r}$ converge in a disc whose radius is the distance to the nearest singular point. It is possible to analytically continue the basis of solutions $\phi(y_0;x)$ to all points in $X$. We can cover $\mathbb{P}^1_{\mathbb{C}}$ by a finite set of open discs $D_{y_k}$ centered at $y_k\in\mathbb{P}^1_{\mathbb{C}}$ such that $\phi(y_k;x)$ converges inside $D_{y_k}$. Since the solutions $\phi(y_k,x)$ and $\phi(y_l,x)$ have to agree for each value of $x$ in the overlapping region $D_{y_k}\cap D_{y_l}$, one can find a matching matrix from the following equation: 
\begin{equation}
	\label{eqn:calcmatchingmatrices}
\phi(y_k;x)=	\begin{pmatrix}\phi_r(y_k;x)\\\vdots\\\phi_1(y_k;x)\\\end{pmatrix}=R_{y_k,y_l}\phi(y_l;x)=R_{y_k,y_l}\begin{pmatrix}\phi_r(y_l;x)\\\vdots\\\phi_1(y_l;x)\\\end{pmatrix}\,.
\end{equation}
Note that the matching matrix $R_{y_k,y_l}$ must be constant. 
Practically, it can be found by numerically evaluating each component of the above equation for several numerical points in the overlapping region. This allows one to determine the matching matrices at least numerically to high precision by taking enough orders in the expansion. In some cases, one may even be able to determine its entries analytically by solving for them in an ansatz for the matrix $R_{y_k,y_l}$. A precise numerical evaluation allows one to identify corresponding analytic expressions in many situations.

\paragraph{The monodromy group.} The Frobenius method allows one to construct a basis of solutions locally for each point $y_0\in \mathbb{P}_{\mathbb{C}}^1$. The local solutions can be analytically continued to a global basis of solutions defined for all $x\in X$. We can also take a point $x\in X$ and a closed loop $\gamma$ starting and ending at $x$ and analytically continue the solution $\phi(y_0;x)$ along $\gamma$. Clearly, if $\gamma$ does not encircle any singular point, Cauchy's theorem implies that the value of $\phi(y_0;x)$ must be the same before and after analytic continuation. One can now ask the question how the vector of solutions $\phi(y_0;x)$ is altered if transported along the small loop $\gamma$ encircling a singular point $x_k$. Let us denote by $\phi_{\circlearrowleft}(y_0;x)$ the value of $\phi(y_0,x)$ after analytic continuation along $\gamma$. Since $\phi_{\circlearrowleft}(y_0;x)$ must still satisfy the differential equation even after analytic continuation, it must be expressible in the original basis $\phi(y_0;x)$, and so there must be a constant $r\times r$ matrix $\rho_{y_0}(\gamma)$ -- called the \emph{monodromy matrix} -- such that
\beq
\phi_{\circlearrowleft}(y_0;x) = \rho_{y_0}(\gamma)\phi(y_0;x)\,.
\eeq
The subindex on $\rho$ denotes the local basis in which the monodromy is
expressed.  Changing the local basis from $y_0$ to $y_1$ amounts to conjugating the monodromy matrix by the matching matrix from eq.~\eqref{eqn:calcmatchingmatrices}:
\beq
\rho_{y_0}(\gamma) =R_{y_0,y_1}\rho_{y_1}(\gamma)R_{y_0,y_1}^{-1}\,.
\eeq

\begin{figure}
	\begin{center}\includegraphics{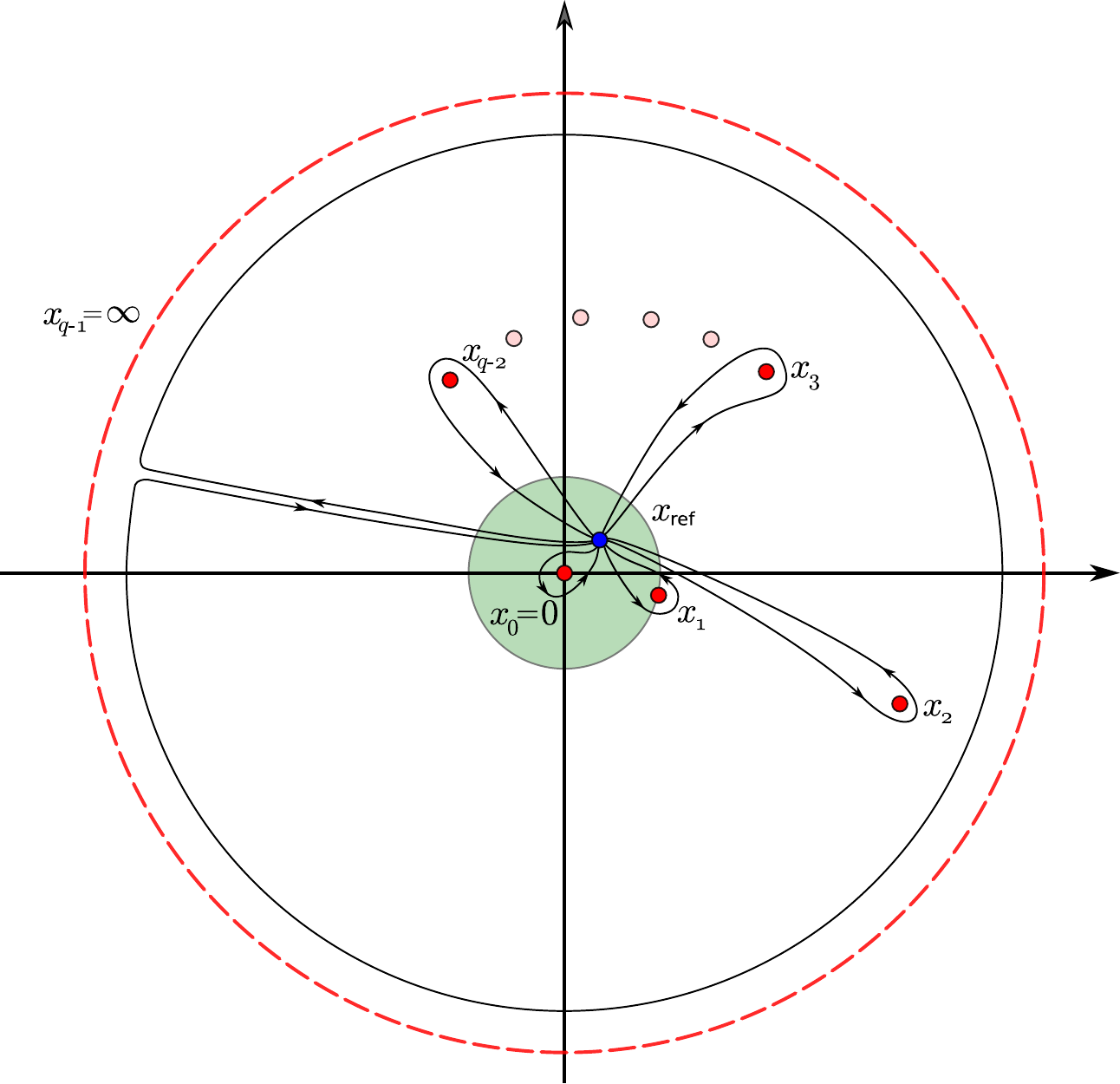}\end{center}
	\caption{Paths for the analytic continuation and the calculation of the monodromies for a differential operator with $q$ regular singular poles, one of which at zero and one at infinity. The (blue) reference point $x_\mathsf{ref}$ has been conveniently chosen in the (green) disc $D_{x_0}$ around $x_0=0$.}
	\label{fig1}
\end{figure}
Let us now explain how we can find the monodromy matrix for a collection of loops $\gamma_{x_k}$ encircling the singular points $x_k$ in the counter-clockwise direction, but no other singular points (see figure~\ref{fig1}). We focus for now on the singular point $x_0$. We can fix a reference point $x_\mathsf{ref} \in D_{x_0}$, and we can also choose the loop $\gamma_{x_0}$ to lie entirely inside $D_{x_0}$. The effect of the analytic continuation on $\phi(x_0;x_\mathsf{ref})$ is easy to describe. Indeed, consider for example the local solution in eq.~\eqref{eq:Frob_log_solution}. Since $\sigma_k(x_0;x)$ is holomorphic at $x_0$, its value does not change when it is analytically continued along $\gamma_{x_0}$. So, only the logarithms $\log(x-x_0)$ and the non-integer powers $(x-x_0)^{s_1}$ are affected. Hence, we find
\beq\bsp
\phi_{i,\circlearrowleft}&(x_0;x_\mathsf{ref}) =\\
& =e^{2\pi i s_1}(x_{\mathsf{ref}}-x_0)^{s_1}\sum_{k=1}^{i} \frac{1}{(k-1)!}\left[\log(x_{\mathsf{ref}}-x_0) + 2\pi i\right]^{k-1}\, \sigma_{k}(x_0;x_\mathsf{ref})\,.
\esp\eeq
In this way, we can work out the entries of the \emph{local} monodromy matrices $\rho_{x_0}(\gamma_{x_0})$ for each singular point $x_0$.

For a singular point $x_k\neq x_0$, we can decompose the loop $\gamma_{x_k}$ based at $x_{\mathsf{ref}}\in D_{x_0}$ into a segment from $x_{\mathsf{ref}}$ to a new reference point $\tilde{x}_{\mathsf{ref}}\in D_{x_k}$, followed by a loop $\tilde{\gamma}_{x_k}$ based at $\tilde{x}_{\mathsf{ref}}$ around $x_k$ and lying entirely inside $D_{x_k}$, and finally we add the segment from $\tilde{x}_{\textrm{ref}}$ to ${x}_{\mathsf{ref}}$ in the opposite direction. Correspondingly, we can then express the monodromy matrix as
\beq\label{eqn:monodromytranslation}
\rho_{x_0}(\gamma_{x_k}) = R_{x_0,x_k}\rho_{x_k}(\tilde{\gamma}_{x_k})R_{x_0,x_k}^{-1}\,,
\eeq
and the local monodromy matrix $\rho_{x_k}(\tilde{\gamma}_{x_k})$ can be determined as described previously. 

Following this procedure, we can associate a monodromy matrix to every singular point $x_k\in\Sigma$.
The set of global monodromy matrices around all singularities but one\footnote{The reason for the monodromy group being generated by one generator less than the number of poles is the following: a loop enclosing no singularity will lead to a trivial monodromy, which is represented as unit matrix. Accordingly, the appropriately ordered product of monodromy matrices with respect to all poles should yield the unit matrix, as the corresponding contour can be deformed into the trivial loop.} expressed in the
basis of the reference neighbourhood will then generate the \emph{monodromy group}.

\paragraph{Mathematical interpretation.} 
%
The differential operator $\cL^{(r)}_x$ determines a \mbox{rank-$r$} vector bundle over $X= \mathbb{P}_{\mathbb{C}}^1\setminus \Sigma$, i.e., for each $x\in X$ the fiber $V_x$ over $x$ is an \mbox{$r$-dimensional} complex vector space, and the solution $\phi(y_0;x)$ is a basis of $V_x$ (because the solutions are linearly independent if $x$ is not a singular point) .

Let $\gamma$ in $X$ be a closed loop based at $x\in X$. The analytic continuation of $\phi(y_0;x)$ along $\gamma$ does not depend on the details of the path. More precisely, the result of the analytic continuation depends on the homotopy class of $\gamma$ in $X$ only. 
Accordingly, it is sufficient to consider the fundamental group $\pi_1(X,x)$. If we fix the basis of solutions $\phi(y_0;x)$, analytic continuation provides a group homomorphism:
\beq\bsp\label{eq:monodromy_rep}
\rho_{y_0} : \pi_1(X,x) &\,\to \textrm{GL}(V_x)\simeq \textrm{GL}_r(\mathbb{C})\,\\
\gamma&\,\mapsto \rho_{y_0}(\gamma)\,.
\esp\eeq
In other words, we can interpret analytic continuation as a representation of the fundamental  group of $X$ of loops based at $x$ in the fiber $V_x$, called the \emph{monodromy representation}. The monodromy group is then the image of $\pi_1(X,x)$ in $\textrm{GL}_r(\mathbb{C})$ under $\rho_{y_0}$. 
In the case of the punctured Riemann sphere $X= \mathbb{P}_{\mathbb{C}}^1\setminus \{x_0,\ldots,x_{q-1}\}$ the structure of the fundamental group is easy to describe: it is the free group generated by the loops $\gamma_{x_k}$, $0\le k<q-1$. Hence, we see that the monodromy group is generated by the matrices $\rho_{y_0}(\gamma_{x_k})$ with $0\le k<q-1$, which are precisely the matrices we have constructed earlier in this section.


\subsection{A class of differential equations allowing for a modular parametrisation} \label{ssec:ClassModularParametrization}
\label{sec:modular_DEQs}
After the general review in the previous subsection, we are now going to describe the class of differential equations we want to discuss in the remainder of this paper. Our starting point is a differential equation of the form~\eqref{eq:DEQ_generic_higher} satisfying the assumptions from section~\ref{sec:FIs_and_deqs}, that is, to all orders in the $\eps$-expansion $\cL_{x,\eps}^{(r)}$ and $N(x,\eps)$ only involve rational functions and MPLs. 
Here, we would like to make the following additional assumptions:
\begin{enumerate}
\item The operator $\cL_x^{(r)}$ is the $(r-1)^{\textrm{th}}$ symmetric power of a degree-two operator $\tilde{\cL}_x^{(2)}$. That is, if the solution space of $\tilde{\cL}_x^{(2)}$ is 
\beq
\textrm{Sol}(\tilde{\cL}_x^{(2)}) = \mathbb{C}\,\psi_1(x)\oplus \mathbb{C}\,\psi_2(x)\,,
\eeq
then the solution space of $\cL_x$ reads 
\beq\label{eq:Sol_L_x}
\textrm{Sol}({\cL}_x^{(r)}) = \bigoplus_{a+b=r-1}\mathbb{C}\,\psi_1(x)^a\psi_2(x)^b\,.
\eeq
\item 
	We make the following assumptions about $\tilde{\cL}_x^{(2)}$. First, we assume that all singular points of $\tilde{\cL}_x^{(2)}$ are MUM-points. We denote the holomorphic and logarithmically-divergent solutions at $x=x_0$ by $\psi_1(x)$ and $\psi_2(x)$ respectively. Second, its monodromy group, which we will call $\Gamma_2$ in the following, is conjugate to a subgroup of $\slz$ of finite index, i.e., there exists $\gamma \in \operatorname{SL}_2(\mathbb C)$ such that $\gamma\Gamma_2\gamma^{-1}$ is a subgroup of $\slz$ of finite index. 
\end{enumerate}
Note that these assumptions imply that the determinant of the Wronskian matrix,
\beq\label{eq:Det_def}
D(x) := \det \left(\begin{smallmatrix}
\psi_1(x) & \psi_2(x) \\
\psi_1'(x) & \psi_2'(x)
\end{smallmatrix}\right)\,,\qquad \psi_a'(x) = \partial_x\psi_a(x)\,,
\eeq
is a rational function of $x$.
While it may seem that these assumptions are rather restrictive, differential equations of this type cover several cases of interesting Feynman integrals. For example, they cover the case of (several) Feynman integrals associated to one-parameter families of elliptic curves ($n=2$) and K3 surfaces~\cite{Doran:1998hm} ($n=3$) where the subtopologies can be expressed in terms of MPLs without additional non-rationalisable square roots. This includes in particular the case of the equal-mass two- and three-loop banana integrals, which are going to be discussed explicitly in section~\ref{sec:bananameromorphic}.
In the remainder of this section, we present a characterisation of the space of functions that is needed to express the result.

\paragraph{The modular parametrisation for $\tilde{\cL}_x^{(2)}$.}
Let us first discuss the structure of the solution space $\textrm{Sol}(\tilde{\cL}_x^{(2)})$. We assume that $x_0=0$ is a MUM-point, and $\psi_1(x)=\phi_1(0;x)$ is holomorphic at $x=0$, while $\psi_2(x)=\phi_2(0;x)$ is logarithmically divergent. We define
\beq\label{eq:tau_def_generic}
\tau := \frac{\psi_2(x)}{\psi_1(x)}\,,\qquad q:= e^{2\pi i \tau}\,.
\eeq
We can always choose a basis of $\textrm{Sol}(\tilde{\cL}_x^{(2)})$ such that $\Im\tau>0$ for $x\in X=\mathbb{P}^1_{\mathbb{C}}\setminus\Sigma$, and so $\tau\in \HP:=\{\tau\in\mathbb{C}:\Im\tau>0\}$. We see that the change of variable from $x$ to $q$ is holomorphic at $x=0$. It can be inverted (at least locally, as a power series) to express $x$ in terms of $q$. This series will converge for $|q|<1$, or equivalently, for all $\tau\in\HP$. It may, however, diverge whenever $x$ approaches a singular point of the differential equation. 

Let us analyse how the monodromy group $\Gamma_2$ acts in the variable $\tau$. Consider $\gamma\in\pi_1(X,x)$. We know that if we analytically continue $\psi(x) = (\psi_2(x),\psi_1(x))^T$ along $\gamma$, then the solution changes to $\psi_{\circlearrowleft}(x) = \tilde{\rho}_0(\gamma)\psi(x) = \left(\begin{smallmatrix}a& b\\c&d \end{smallmatrix}\right)\psi(x)$, for some $\left(\begin{smallmatrix}a& b\\c&d \end{smallmatrix}\right)\in \Gamma_2\subseteq \slz$. It is then easy to see that the monodromy group acts on $\tau$ via M\"obius transformations:
\beq\label{eq:Moebiusaction}
\tau_{\circlearrowleft} = \frac{a\tau+b}{c\tau+d} =: \gamma\cdot \tau\,.
\eeq
Clearly, $x$ should not change under analytic continuation (because $x$ is a rational function, and thus free of branch cuts), and so $x(\tau)$ must be invariant under the action of the monodromy group:
\beq\label{eq:modular_functions_def}
x\left( \frac{a\tau+b}{c\tau+d}\right) = x(\tau) \,, \textrm{ for all } \left(\begin{smallmatrix}a& b\\c&d \end{smallmatrix}\right)\in \Gamma_2\,.
\eeq
A (meromorphic) function from $\HP$ to $\mathbb{C}$ that satisfies eq.~\eqref{eq:modular_functions_def} is called a \emph{modular function} for  $\Gamma_2$. If we define $h_1(\tau) := \psi_1(x(\tau))$, then $h_1$ changes under analytic continuation according to:
\beq\bsp\label{eq:h1_transform}
h_1\left( \frac{a\tau+b}{c\tau+d}\right) &\,= h_1(\tau)_{\circlearrowleft} = \psi_1(x(\tau))_{\circlearrowleft} \\
&\,= c\,\psi_2(x(\tau))+d\, \psi_1(x(\tau)) = (c\tau+d)\,h_1(\tau)\,.
\esp\eeq
A holomorphic function from $\eHP := \HP\cup \mathbb{P}^1_{\mathbb{Q}}$ to $\mathbb{C}$ that satisfies eq.~\eqref{eq:h1_transform} is called a \emph{modular form of weight 1} for  $\Gamma_2$ (see section~\ref{sec:modular_review}). 
We see that whenever $\Gamma_2\subseteq \slz$, the differential equation $\tilde\cL_x^{(2)}u(x)=0$ admits a \emph{modular parametrisation}, by which we mean that there is a modular function $x(\tau)$ and a modular form $h_1(\tau)$ of weight 1 for $\Gamma_2$ such that
\beq
\textrm{Sol}(\tilde\cL_x^{(2)}) = h_1(\tau)\big(\mathbb{C} \oplus \mathbb{C}\tau\big)\,.
\eeq

\paragraph{Mathematical interpretation.} The solutions of $\tilde\cL_x^{(2)}$ define multivalued holomorphic functions on $X$. We can ask the question: On which surface these functions are single-valued holomorphic functions? This can be realised when expressing the solutions in the new variable $\tau\in\HP$. The monodromy group $\Gamma_2\subset \textrm{GL}_2(\mathbb{C})$ associated to the differential operator acts on $\HP$ via M\"obius transformations. We can identify the space on which $\psi_1(x(\tau))=h_1(\tau)$ is holomorphic and single-valued with $\HP$. Let us mention, however, that the action of $\Gamma_2$ on $\HP$ factors through its projection $\overline{\Gamma}_2$ on $\textrm{PGL}_2(\mathbb{C})$, where we have identified matrices that only differ by a non-zero multiplicative constant. Indeed, it is easy to see that $\left(\begin{smallmatrix}a & b\\ c& d\end{smallmatrix}\right) \in \textrm{GL}_2(\mathbb{C})$ and $\lambda\left(\begin{smallmatrix}a & b\\ c& d\end{smallmatrix}\right) \in \textrm{GL}_2(\mathbb{C})$ lead to the same M\"obius transformation in eq.~\eqref{eq:Moebiusaction}, for all $\lambda\in \mathbb{C}^*$. The action on $h_1(\tau)$, however, may be sensitive to $\lambda$.

Different points $\tau$ in $\HP$ correspond to the same value of $x$ in our original space $X$, and the points that are identified are precisely those related by the action of the monodromy group $\Gamma_2$. It is thus natural to consider the space $Y_{\Gamma_2}=\Gamma_2\backslash\HP$. The function $x(\tau)$ defines a holomorphic map from $\HP$ to $X$, and it is a bijection between  $Y_{\Gamma_2}$ and $X$. The punctured Riemann sphere can be compactified to $\overline{X}\simeq \mathbb{P}^1_{\mathbb{C}}$ by adding the singularities. Similarly, we can compactify $Y_{\Gamma_2}$ to the space $X_{\Gamma_2} = \Gamma_2\backslash \eHP$, with $\eHP:=\HP\cup\mathbb{P}^1_\mathbb{Q}$ the extended upper half-plane. The pre-images of the singular points at the orbits $\Gamma_2\backslash \mathbb{P}^1_{\mathbb{Q}}$ are called the \emph{cusps} of $X_{\Gamma_2}$ (see section~\ref{sec:modular_review}).

Let us finish this interlude by mentioning that $X_{\Gamma_2}$ and $Y_{\Gamma_2}$ are not manifolds, but \emph{orbifolds}. Loosely speaking, an $n$-dimensional manifold is a topological space that locally `looks like' $\mathbb{R}^n$. Similarly, an $n$-dimensional orbifold locally looks like a quotient $\Gamma_2\backslash \mathbb{R}^n$. This has a bearing on how we choose coordinates on $X_{\Gamma_2}$ and $Y_{\Gamma_2}$. Indeed, the chosen coordinate in a neighbourhood of $\tau\in\eHP$ will depend on whether $\tau$ has a non-trivial stabilizer $(\Gamma_2)_{\tau} = \{\gamma\in\Gamma_2: \gamma\cdot\tau=\tau\}$. We will discuss this in more detail in section~\ref{sec:modular_review}.

\paragraph{The modular parametrisation for ${\cL}_x^{(r)}$.} Since the solution spaces of $\tilde{\cL}_x^{(2)}$ and ${\cL}_x^{(r)}$ are related, it is not surprising that all the symmetric powers of $\tilde{\cL}_x^{(2)}$ will also admit a modular prarametrisation. If we define $\tau$ again by eq.~\eqref{eq:tau_def_generic}, we have
\beq
\textrm{Sol}(\cL_x^{(r)}) = h_1(\tau)^{r-1}\,\bigoplus_{s=0}^{r-1}\mathbb{C}\tau^s\,.
\eeq
Since the elements of $\textrm{Sol}(\cL_x^{(r)})$ are the maximal cuts of the Feynman integral $I(x,0)$ in $d=d_0$ dimensions, we see that the maximal cuts are linear combinations of a modular form of weight $r-1$ for $\Gamma_2$, multiplied by additional powers of $\tau$. More generally, $h_1(\tau)^{r-1}$ is also a modular form of weight $r-1$ for the monodromy group of $\cL_x^{(r)}$. 
The maximal cuts of the other master integrals are obtained by differentiation. Using
\beq
\label{eqn:inverseJacobian}
\partial_x = \frac{\mathcal{D}(\tau)}{h_1(\tau)^2}\,\partial_{\tau}\,, \qquad \mathcal{D}(\tau):=D(x(\tau))\,,
\eeq
we see that the maximal cuts of the other master integrals also involve the derivatives of $h_1(\tau)$ (with respect to $\tau$). As we will see in the next section, the latter are no longer modular forms, but they give rise to so-called quasi-modular forms.

Let us now return to the original inhomogeneous differential equation. To solve this equation in terms iterated integrals, we can turn it into a first-order inhomogeneous system for the vector $\cI(x,\eps)$ and proceed similar to  section~\ref{sec:FIs_and_deqs}. The entries of the Wronskian matrix of $\cL_x^{(r)}$ can be expressed in terms of $\psi_1(x)$ and $\psi_2(x)$:
\begin{align}
&W_r(x)_{ij} = \\
\nonumber&\,=\binom{r-1}{i-1}^{-1}\,\sum_{k=0}^{j-1}\binom{r-j}{i-k-1}\binom{j-1}{k} \psi_1(x)^{r-i-j+k+1}\,\psi_2(x)^{j-k-1}\psi_1'(x)^{i-k-1}\,\psi_2'(x)^{k}\,,
\end{align}
with determinant
\beq
\det W_r(x) = D(x)^{r(r-1)/2}\,\prod_{k=1}^{r-1}\frac{k!}{k^k}\,.
\eeq
Note that $W_r(x)$ is rational whenever $D(x)$ is.
The iterated integrals that arise from expanding the path-ordered exponential in eq.~\eqref{eq:path_ordered} will involve differential one-forms of the form
\beq\label{eq:diff_form_sample}
dx\,R(x)\,\psi_1(x)^{\alpha}\,\psi_2(x)^{\beta}\psi_1'(x)^{\gamma}\,\psi_2'(x)^{\delta}\,,
\eeq
where $R(x)$ is a rational function and $\alpha$, $\beta$, $\gamma$, $\delta$ are positive integers. For applications, one is often interested in knowing a basis of special functions and associated differential forms by integration of which all the iterated integrals can be built. In the case $\alpha=\beta=\gamma=\delta=0$, the answer to this question is well known, and the corresponding basis of special functions are the multiple polylogarithms in eq.~\eqref{eq:MPL_def}. In the case where at least one of the exponents is non-zero, we can change variables to $\tau$. The Jacobian is (cf.~eq.~\eqref{eqn:inverseJacobian}) 
\beq\label{eq:jacobian}
dx = \frac{h_1(\tau)^2}{\mathcal{D}(\tau)}\,d\tau\,.
\eeq
Since $D(x)$ is a rational function, we can eliminate $\psi_2'(x)$. We can also eliminate $\psi_2(x)$ in favour of $\psi_1(x)$ and $\tau$. Hence, it is sufficient to consider differential forms of the form
\beq\label{eq:diff_form}
d\tau\,R(x(\tau))\,h_1(\tau)^m\,h_1'(\tau)^s\,\tau^p\,,
\eeq
where $m,s,p\in\mathbb{Z}$, with $s,p$ positive. We can write 
\beq
\frac{1}{p!}\tau^p = \int_{i\infty}^\tau d\tau_1\int_{i\infty}^{\tau_1}d\tau_2\cdots \int_{i\infty}^{\tau_{p-1}}d\tau_p\,,
\eeq
where the divergence at $i\infty$ is regulated by interpreting the lower integration boundary as a tangential base-point~ \cite{Brown:mmv}. It is therefore sufficient to consider differential forms with $p=0$.
One of the main tasks of the remainder of this paper is to answer this question for the class of iterated integrals in eq.~\eqref{eq:diff_form_sample}. More precisely, we will give a constructive proof of the following result.
\begin{thm}\label{thm:section2}
	With assumptions and notations as in section \ref{ssec:ClassModularParametrization}, at every order in $\eps$, the solution of the differential equation~\eqref{eq:DEQ_generic_higher} can be written as a $\mathbb{C}$-linear combination of iterated integrals of meromorphic modular forms for the monodromy group $\Gamma_2$.
\end{thm}
We will give a constructive proof of Theorem \ref{thm:section2} in section~\ref{sec:mero_sec}. In addition, in section~\ref{sec:mero_sec} we will completely classify the relevant iterated integrals and give an explicit basis.

\vspace{1mm}\noindent



\section{Review of (quasi-)modular forms and their iterated integrals}
\label{sec:modular}
The previous section has shown that there are certain classes of Feynman integrals whose differential equations admit a modular parametrisation. This is to say that their maximal cuts in $D=d_0$ dimensions are linear combinations of derivatives of modular forms multiplied by powers of $\tau$, and the higher orders in $\eps$ of the maximal cuts and the full uncut integral can be expressed in terms of iterated integrals of such functions. The aim of this section is to briefly review the theory of holomorphic modular forms and their iterated integrals. In the next section we will extend this to include iterated integrals of meromorphic modular forms.

\subsection{The modular group $\slz$ and its subgroups}
\label{sec:modular_review}
We start by reviewing some general facts about (certain) subgroups of the modular group $\slz$. For a review, see ref.~\cite{diamond2005first}, and references therein. Let $\Gamma$ denote a subgroup of $\slz$ of finite index, i.e., the quotient $\Gamma\backslash\slz$ is finite (which means, intuitively, the we can cover $\slz$ by a finite number of copies of $\Gamma$). In the following we denote the index of $\Gamma$ in $\slz$ by
\beq
[\slz:\Gamma] = \left|\Gamma\backslash\slz\right| < \infty\,.
\eeq
An important example of finite-index subgroups are the \emph{congruence subgroups of level $N$}, with $N$ a positive integer, defined as those subgroups $\Gamma$ that contain the principal congruence subgroups $\Gamma(N) = \{\abcd \in \slz \, :\abcd = \left(\begin{smallmatrix}1&0\\0&1\end{smallmatrix}\right)\bmod N\}$. An important example of congruence subgroup are the groups 
\beq\label{eq:Gamam1(N)_def}
\Gamma_1(N) := \{\abcd \in \slz \, :\abcd = \left(\begin{smallmatrix}1&*\\0&1\end{smallmatrix}\right)\bmod N\}\,.
\eeq
In the following we will keep the discussion general, and we do not restrict ourselves to congruence subgroups, unless specified otherwise.

The modular group and its subgroups naturally act on the extended upper half-plane $\eHP = \HP\cup \mathbb{P}^1_\mathbb{Q}$ by M\"obius transformations via
\begin{equation}
  \label{eq:modulartrafo}
  \gamma \cdot \tau = \frac{a\tau + b}{c\tau + d},\quad \gamma = \abcd\in\slz.
\end{equation}
$\Gamma$ acts separately on $\HP$ and
$\bP^1_\bQ$, and decomposes $\bP^1_\bQ$ into disjoint orbits, called \emph{ cusps of }$\Gamma$:\footnote{By abuse of language, one also often calls the elements of $\bP^1(\bQ)$ cusps.}
\begin{equation}
	S_\Gamma:=\Gamma\setminus\bP^1_\bQ\,.
\end{equation}
The number of cusps of $\Gamma$ is always finite and we denote it by $\eps_{\infty}(\Gamma) := \#S_{\Gamma} < \infty$. The stabilizer of a cusp $s\in\mathbb{P}^1_\mathbb{Q}$ is generated by an element of the form $\pm T^{h} = \pm\left(\begin{smallmatrix} 1& h\\0&1\end{smallmatrix}\right)$, for some integer $h$ called the \emph{width} of the cusp. In case the stabilizer of the cusp $s$ contains an element $-T^h\in\Gamma_{s}$, the cusp is called \emph{irregular}, otherwise it is \emph{regular}. The numbers of regular and irregular cusps are denoted by $\eps_r(\Gamma)$ and $\eps_i(\Gamma)$ respectively. Note that for every cusp $s\in \mathbb{Q}$ there exists a $\gamma\in\slz$ such that $\gamma \cdot s=i\infty$.

A point $\tau\in\fH$ is called an \emph{elliptic point for $\Gamma$} if $\tau$ has a non trivial stabilizer group in $\Gamma$:
\beq
\Gamma_{\tau} := \{\gamma\in\Gamma: \gamma\cdot \tau = \tau\}\,.
\eeq
One can show that $\Gamma_\tau$ is always a finite-cyclic group. If $\Gamma_{\tau}$ is cyclic of order $n$, then $\tau$ is called an elliptic point of order $n$. $\slz=\Gamma(1)$ has exactly two elliptic points, $i$ and $\rho :=  e^{2\pi i/3}$ in its fundamental domain $\mathcal{D}_1\cup \mathcal{D}_2\cup \mathcal{D}_3$, with
\beq\bsp
\mathcal{D}_1 &\,:= \Big\{\tau\in \HP: |\tau|>1\textrm{ and } |\Re\tau|<{\frac{1}{2}}\Big\}\,,\\
\mathcal{D}_2 &\,:=  \Big\{\tau\in \HP: |\tau|\ge1\textrm{ and } \Re\tau={-\frac{1}{2}}\Big\}\,,\\
\mathcal{D}_3 &\,:=  \Big\{\tau\in \HP: |\tau|=\textrm{ and } {-\frac{1}{2}}<\Re\tau\ge 0\Big\}\,.
\esp\eeq
They are of order two and three respectively,
\beq
 \Gamma_i \simeq \mathbb{Z}/2\mathbb{Z} \text{~~~and~~~}  \Gamma_{\rho} \simeq \mathbb{Z}/3\mathbb{Z}\,,\qquad \rho := e^{2\pi i/3}\,.
\eeq
Every elliptic point is $\slz$-equivalent to either $i$ or $\rho:=e^{2\pi i/3}$. The number of elliptic points of order two or three of $\Gamma$ will be denoted by $\eps_2(\Gamma)$ and $\eps_3(\Gamma)$. The principal congruence subgroups $\Gamma(N)$ have no elliptic points for $N>1$. The subgroups $\Gamma_1(N)$ have no elliptic points for $N>3$, while $\Gamma_1(3)$ has no elliptic points of order two and $\Gamma_1(2)$ has no elliptic points of order three.

\subsection{Modular curves}
The space of orbits $X_{\Gamma} := \Gamma\backslash\eHP$ can be equipped with the structure of a compact Riemann surface, called the \emph{modular curve for $\Gamma$}. 
The genus of $\Gamma$ is defined as the genus of $X_{\Gamma}$ and is related to the number of cusp and elliptic points of $\Gamma$:
\begin{equation}\label{eq:genus}
g = 1+d_{\Gamma}-\frac{\eps_2(\Gamma)}{4}-\frac{\eps_3(\Gamma)}{3}-\frac{\eps_{\infty}(\Gamma)}{2}\,,
\end{equation}
where we introduced the shorthand $d_{\Gamma} := \frac{[\slz:\{\pm1\}\Gamma]}{12}$. 
In the remainder of this paper we are only concerned with the case where $\Gamma$ has genus zero. It is known that $\Gamma_1(N)$ an $\Gamma(N)$ have genus zero for $N\le12$ and $N\le5$ respectively. In particular, the group $\Gamma_1(6)$ relevant to the equal-mass sunrise and banana graphs has genus zero. A complete list of all genus zero subgroups can be found in refs.~\cite{YifanYang,allgenus0}.

In the following it will be important to know how we can define local coordinate charts on the Riemann surface $X_{\Gamma}$. We recall that $X_{\Gamma}$ is an orbifold, and the points of $X_{\Gamma}$ are equivalence classes $[\tau] = \{\gamma\cdot \tau:\gamma\in\Gamma\}$. Let $P=[\tau_0]\in X_{\Gamma}$. To define a local coordinate $z$ such that $z(P)=0$ in a neighbourhood of $P$, we need to distinguish three cases:
\begin{itemize} 
\item If $\tau_0$ is an elliptic point of order $h$, a local coordinate is defined by $z=(\tau-\tau_0)^h$.
\item If $\tau_0$  is a cusp of width $h$, such that $\gamma\cdot \tau_0=i\infty$, a local coordinate is defined by $z= e^{2\pi i(\gamma\cdot\tau)/h'}$, with $h'=h$ is $\tau_0$ is a regular cusp, and $h'=2h$ otherwise.
\item If $\tau_0$ is neither a cusp nor an elliptic point, $z=\tau-\tau_0$ is a good local coordinate.
\end{itemize}

The field of meromorphic functions of $X_{\Gamma}$ is isomorphic to the field $\cM_0(\Gamma)$ of modular functions, i.e., meromorphic functions $f:\eHP\to \bC$ that satisfy
\begin{equation}\label{eq:modular_function}
   f\left(\frac{a\tau+b}{c\tau+d}\right) = f(\tau)\,,\qquad \forall \left(\begin{smallmatrix} a& b\\c& d\end{smallmatrix}\right)\in \Gamma\,.
\end{equation}
For every meromorphic function, we denote by $\nu_P(f)\in \mathbb{Z}$ the \emph{order of vanishing at $P$}, i.e., $\nu_P(f)>0$ ($<0$) if $f$ has a zero (pole) of order $|\nu_P(f)|$ at $P$. If $z$ denotes the local coordinate introduced above, we have $f(\tau) = A\,z^{\nu_P(f)} + \mathcal{O}(z^{\nu_P(f)+1})$, with $A\neq 0$.

If $X_{\Gamma}$ has genus zero, the field of meromorphic functions on $X_{\Gamma}$ has a single generator, $\cM_0({\Gamma})\simeq \mathbb{C}(\tz)$, for some $\tz\in \cM_0({\Gamma})$ called a \emph{Hauptmodul}. Every modular function is a rational function in the Hauptmodul $\tz$. If $h$ is the width of the infinite cusp, then we can choose the Hauptmodul to have the $q$-expansion~\cite{YifanYang}
\beq \label{eqn:Hauptmodul}
\tz(\tau) = q^{-1/h} + \sum_{n\ge 0} a_0\,q^{n/h}\,,\qquad q= e^{2\pi i\tau}\,.
\eeq
In the following we always assume that such a Hauptmodul $\tz$ has been fixed.

\subsection{Review of (quasi-)modular forms}
\label{ssec:mfcs}

\subsubsection{Meromorphic modular forms}
Let $k$ be an integer,  $\Gamma\subseteq \slz$.
We define the action of weight $k$ of $\Gamma$ on a function $f:\eHP\to \bC$ by 
\beq
f[\gamma]_k(\tau) := (c\tau+d)^{-k}\,f(\gamma\cdot\tau)\,,\qquad \gamma=\left(\begin{smallmatrix}a&b\\c&d\end{smallmatrix}\right)\in\Gamma\,.
\eeq
 A weakly modular form of weight $k$ for $\Gamma$ is a function that is invariant under this $\Gamma$-action,
 \beq
 \label{eq:defmf}
f[\gamma]_k(\tau) = f(\tau)\,,\qquad \forall \gamma\in\Gamma\,.
 \eeq
A \emph{meromorphic modular form of weight $k$ for $\Gamma$} is a weakly modular form $f$ of weight $k$ for $\Gamma$ that is meromorphic on $\HP$ and at every cusp, i.e., it admits a $q$-expansion of the form
\beq
f[\gamma]_k(\tau) = \sum_{n\ge n_0}a_n\,q^{n/h}\,, \qquad \forall\gamma=\left(\begin{smallmatrix}a&b\\c&d\end{smallmatrix}\right)\in\slz\,,
\eeq
where $h$ is the width of the cusp $\frac{a}{c}$. We denote the $\bC$-vector space of meromorphic modular forms of weight $k$ for $\Gamma$ by $\cM_k(\Gamma)$, and we write $\cM(\Gamma) := \bigoplus_k\cM_k(\Gamma)$. In particular, $\cM_0(\Gamma)$ is the field of modular functions for $\Gamma$ (see eq.~\eqref{eq:modular_function}).
Holomorphic modular forms are defined in an analogous manner. The $\bC$-vector space of holomorphic modular forms of weight $k$ for $\Gamma$ is denoted by $M_k(\Gamma)$, and we define $M(\Gamma) := \bigoplus_kM_k(\Gamma)$. Note that $M_k(\Gamma)$ is always finite-dimensional, and $\dim_{\bC} M_k(\Gamma)=0$ for $k\le 0$. In the following we refer to holomorphic modular forms simply as modular forms.

A (meromorphic) \emph{cusp form} is a (meromorphic) modular form for which $a_0=0$ for every cusp. We denote the vector space of (meromorphic) cusp forms of weight $k$ by $S_k(\Gamma)$ ($\cS_k(\Gamma)$).  The space of cusp forms $S_k(\Gamma)$ is an ideal in $M_k(\Gamma)$. The quotient is the space of Eisenstein series $E_k(\Gamma)$, and there is a direct sum decomposition
\beq
M_k(\Gamma) = E_k(\Gamma) \oplus S_k(\Gamma)\,.
\eeq

\subsubsection{Meromorphic quasi-modular forms}
In general, the derivative of a (meromorphic) modular form is no longer a modular form, but we need to introduce a more general class of functions. A \emph{meromorphic quasi-modular form of weight $k$ and depth $p$ for $\Gamma$} is a function $f:\eHP\to\bC$ that is meromorphic on $\HP$ and at the cusps, and it transforms according to 
\begin{equation}\label{eq:quasi_modular}
f[\gamma]_k(\tau) = \sum_{r=0}^{p}f_r(\tau)\left(\frac{c}{c\tau+d}\right)^r\,,\quad\gamma=\abcd\in\Gamma\,,
\end{equation}
where the $f_0,\ldots,f_p$ are meromorphic functions, with $f_p\neq 0$. Note that eq.~\eqref{eq:quasi_modular} for $\gamma=\operatorname{id}$ implies that $f_0=f$.
The $\bC$-vector space of meromorphic quasi-modular forms of weight $k$ and depth at most $p$ is denoted by $\cQ\cM_k^{\le p}(\Gamma)$. Quasi-modular forms of depth zero are precisely the modular forms. Holomorphic quasi-modular forms are defined in an analogous way, and the corresponding (finite-dimensional) vector space is denoted by $QM_k^{\le p}(\Gamma)$. Note that $\dim_{\bC}QM_k^{\le p}(\Gamma)=0$ for $k\le 0$ and $2p>k$. We also use the notations 
\beq
QM_k(\Gamma) := \bigcup_{p\ge 0}QM_k^{\le p}(\Gamma) \textrm{~~and~~} QM(\Gamma) := \bigoplus_{k}QM_k(\Gamma)\,.
\eeq
The vector spaces $\cQ\cM_k(\Gamma)$ and $\cQ\cM(\Gamma)$ are defined in a similar fashion.

The algebra of all (meromorphic) quasi-modular forms is closed under differentiation. We use the notation $\delta:=\frac{1}{2\pi i} \partial_{\tau}= q\,\partial_q$. If $f$ is a quasi-modular form of weight $k$ and depth at most $p$, then $\delta f$ has weight $k+2$ and depth at most $p+1$.

The Eisenstein series $G_2(\tau)$ of weight two is the prime example of a quasi-modular form of depth 1. The Eisenstein series are defined as\footnote{For $k=1$ the series is not absolutely convergent. Here we assume the standard summation convention for $G_2(\tau)$, cf., e.g., ref.~\cite{diamond2005first}} 
\beq
G_{2k}(\tau) = \sum_{(m,n)\in\bZ^2\setminus(0,0)}\frac{1}{(m\tau+n)^{2k}}\,.
\eeq
For $k\neq 1$, $G_{2k}(\tau)$ is a modular form of weight $2k$ for $\slz$. For $k=1$, we have
\begin{equation}\label{eq:G2_transform} 
G_2[\gamma]_2(\tau) =  (c\tau+d)^{-2} G_2(\gamma\cdot\tau)=G_2(\tau)-\frac{1}{4\pi i}\frac{c}{c\tau+d}\,,
\end{equation}
for every $\gamma= \abcd\in\slz$. 
Hence $G_{2}(\tau)$ defines a (holomorphic) quasi-modular form of weight 2 and depth 1 for $\slz$. In fact, one can show that every meromorphic quasi-modular form of depth $p$ can be written as a polynomial of degree $p$ in $G_2(\tau)$:
\beq
\cQ\cM(\Gamma) = \cM(\Gamma)[G_2(\tau)]\,.
\eeq


\subsection{Iterated integrals of holomorphic quasi-modular forms}
The previous discussion makes it clear that the functions $f(\tau) := R(x(\tau))\,h_1(\tau)^m\,h_1'(\tau)^s$ in eq.~\eqref{eq:diff_form} (with $p=0$) are a meromorphic quasi-modular forms of weight $m+3s$ and depth at most $s$. Hence, we see that the iterated integrals encountered at the end of section~\ref{sec:DEQs} are iterated integrals of meromorphic quasi-modular forms for the monodromy group $\Gamma_2$. In the remainder of this section we give a short review of iterated integrals of holomorphic quasi-modular forms, following refs.~\cite{ManinModular,Brown:mmv}. In the next section we present the extension to the meromorphic case.

Let $h_1,\ldots, h_k$ be meromorphic quasi-modular forms for $\Gamma\subseteq \slz$ We define their iterated integral by~\cite{ManinModular,Brown:mmv}
\beq\label{eq:II_def}
I(h_1,\ldots, h_k;\tau) := \int_{i\infty}^\tau d{\tau_1}\,h_1(\tau_1)\int_{i\infty}^{\tau_1} d{\tau_2}\,h_2(\tau_2)\int_{i\infty}^{\tau_2}\cdots\int_{i\infty}^{\tau _{k-1}}d{\tau_k}\,h_k(\tau_k)\,.
\eeq
At this point we have to mention that this definition requires a regularisation of the divergence at  $\tau=i\infty$, already in the holomorphic case (at least for Eisenstein series). For the holomorphic case we follow ref.~\cite{Brown:mmv}, and interpret the lower integration boundary as a tangential base point. In the meromorphic case, the regularisation requires the use of tools from renormalisation theory, see ref.~\cite{matthes2021iterated}. We refer to refs.~\cite{Brown:mmv,matthes2021iterated} for a detailed discussion. 

The iterated integrals in eq.~\eqref{eq:II_def} are not necessarily independent, even if the $h_1,\ldots, h_k$ are linearly independent in $\cQ\cM(\Gamma)$. Rather, we have to identify a set of quasi-modular forms that are linearly independent up to total derivatives, i.e., modulo $\delta \cQ\cM(\Gamma)$ (see also the discussion in section~\ref{sec:main_thm}). 
Said differently, we need to would like to know which quasi-modular forms can be expressed as derivatives of (quasi-)modular forms. This question can be answered completely in the holomorphic case. Writing $QM_k(\Gamma) := \bigcup_{p\ge 0}QM_k^{\le p}(\Gamma)$, we have the decomposition~\cite{,ZagierModular,AMBP_2012__19_2_297_0}
\beq\label{eq:holomorphic_decomposition}
QM_k(\Gamma) = \left\{\begin{array}{ll}
\bC\,,& k = 0\,,\\
M_1(\Gamma)\,,& k=1\,,\\
M_2(\Gamma)\oplus\bC\,G_2(\tau)\,,&k=2\,,\\
M_k(\Gamma)\oplus\delta QM_{k-2}(\Gamma)\,,&k>3\,.
\end{array}\right.
\eeq
Note that the sums are direct, i.e., every holomorphic quasi-modular form of weight $k>2$ can be written modulo derivatives as a holomorphic modular form, and this decomposition is unique. The decomposition can be performed in an algorithmic way, cf.~ref.~\cite{AMBP_2012__19_2_297_0}.
In other words, modulo derivatives, $QM(\Gamma)$ is generated as a vector space by $M(\Gamma)$ and $G_2(\tau)$. Consequently, in the holomorphic case it is sufficient to consider iterated integrals of modular forms and $G_2(\tau)$~\cite{Matthes:QuasiModular}. The equivalent of the decomposition in eq.~\eqref{eq:holomorphic_decomposition} in the meromorphic case for general subgroups $\Gamma$ is currently still unknown, and results are only available for \emph{weakly holomorphic modular forms} (i.e., with poles at most at the cusps)~\cite{MR2407067} and for meromorphic quasi-modular forms for the whole modular group, $\Gamma=\slz$~\cite{matthes2021iterated}. One of the main results of this paper is the generalisation of eq.~\eqref{eq:holomorphic_decomposition} and the results of ref.~\cite{matthes2021iterated} to arbitrary subgroups of genus zero.

\section{Iterated integrals of meromorphic modular forms}
\label{sec:mero_sec}

\subsection{A decomposition theorem for meromorphic quasi-modular forms}
\label{sec:main_thm}

In this section we state and prove the generalisation of eq.~\eqref{eq:holomorphic_decomposition} for all genus-zero subgroups. The special case $\Gamma=\slz$ was proved by one of us in ref.~\cite{matthes2021iterated}, and the proof presented here is a generalisation of that proof. Before we state the main theorem in this section, we need to introduce some notation. Let $R\subset X_{\Gamma}\setminus S_{\Gamma}$ be a finite set of points which are not cusps, and let $s_0\in S_{\Gamma}$ be a cusp of $X_{\Gamma}$, and $R_{s_0} := R\cup \{s_0\}$ and $R_S:=R\cup S_{\Gamma}$. We define $\cM_k(\Gamma,R_{s_0})$ to be the sub-vector space of $\cM_k(\Gamma)$ consisting of all meromorphic modular forms of weight $k$ for $\Gamma$ with poles at most at points in $R_{s_0}$.

\begin{defi}\label{defi:Mtilde}
Define $\widetilde{\cM}_k(\Gamma,R_{s_0})$ to be the subset of those $f\in\cM_k(\Gamma,R_{s_0})$ which satisfy: 
\begin{enumerate}
\item $\nu_P(f)\ge\frac{1-k}{h_P}$, for all $P\in R$;
\item $\nu_s(f)\ge 0$, for all $s\in S_{\Gamma}\setminus\{s_0\}$;
\item $\lfloor\nu_{s_0} (f)\rfloor \ge -\dim S_k(\Gamma)$.
\end{enumerate}
\end{defi}
\noindent In the previous definition $\lfloor x \rfloor$ is the floor function, i.e., the largest integer less or equal than $x\in \mathbb{R}$.

\begin{thm}\label{thm:main}
Let $\Gamma\subseteq \slz$ have genus zero, $R_{s_0}$ as defined above. We have a decomposition
\beq\nonumber
\cQ\cM_k(\Gamma,R_S) = \left\{\begin{array}{ll}
 \delta\cQ\cM_{k-2}(\Gamma,R_S) \oplus \cM_{k}(\Gamma,R_S)\,, & \text{ for $k<2$}\,,\\
\displaystyle \delta\cQ\cM_{k-2}(\Gamma,R_S) \oplus \cM_{2-k}(\Gamma,R_S)\,G_2^{k-1}\oplus \widetilde{\cM}_k(\Gamma,R_{s_0})\,,&\text{ for $k\ge 2$}\,.
\end{array}\right.
\eeq
\end{thm}

The proof is presented in appendix~\ref{app:mathy}. Theorem~\ref{thm:main} generalises the result of ref.~\cite{matthes2021iterated} to arbitrary subgroups of genus zero.
In section~\ref{sec:neat_proof} we sketch the proof for a subset of subgroups of genus zero, the so-called \emph{neat} subgroups (see Definition~\ref{defi:neat}). The proof of section~\ref{sec:neat_proof} is constructive, and allows one to perform the decomposition in Theorem~\ref{thm:main} explicitly for neat subgroups. We expect that this case covers most of the applications to Feynman integrals. Before we discuss the proof for neat subgroups, however, we review some consequences of Theorem~\ref{thm:main}.

\paragraph{Proof of Theorem~\ref{thm:section2}.}
We now show that the decomposition in Theorem~\ref{thm:main} immediately leads to a proof of Theorem~\ref{thm:section2}. We have already argued that the class of differential equations considered in section~\ref{sec:modular_DEQs} leads to iterated integrals involving the one-forms in eq.~\eqref{eq:diff_form} with $p=0$, and the functions $f(\tau) := R(x(\tau))\,h_1(\tau)^m\,h'_1(\tau)^s$ are quasi-modular forms of weight $k:=3s+m$ and depth at most $s$ for the monodromy group $\Gamma_2$. Let $f$ have poles at most at the cusps and at some finite set of points $R\subset \HP$, i.e., $f\in \cQ\cM_{k}^{\le s}(\Gamma_2,R_S)$. Theorem~\ref{thm:main} then implies that, for some fixed choice of cusp $s_{0}\in S_{\Gamma_2}$:
\begin{itemize}
\item If $k<2$, there is $h\in \cM_{k}(\Gamma_2,R_S)$ and $g\in\cQ\cM_{k-2}(\Gamma_2,R_S)$ such that $f = h+\delta g$. 
\item If $k\ge 2$, there are modular forms $\tilde{h}\in \widetilde{\cM}_{k}(\Gamma_2,R_{s_0})$, ${h}\in {\cM}_{2-k}(\Gamma_2,R_S)$ and a quasi-modular form $g\in\cQ\cM_{k-2}(\Gamma_2,R_S)$ such that $f = \tilde{h}+{h}\,G_2^{k-1}+\delta g$. 
\end{itemize}
The derivatives $\delta g$ can be trivially integrated away, and we see that we only need to consider the meromorphic modular form $\tilde{h}$ or the quasi-modular form ${h}\,G_2^{k-1}$, the latter being characterised by the fact that it has the maximally allowed depth, $s=k-1$. In order to show that Theorem~\ref{thm:section2} holds, we need to show that these quasi-modular forms with maximally allowed depth $s=k-1$ do not arise from our class of differential equations. To see this, we start from eq.~\eqref{eq:diff_form_sample} with $\alpha$, $\beta$, $\gamma$, $\delta$ positive integers, and we change variables from $x$ to $\tau$, and we trace the powers of $G_2$ that are produced along the way. The Jacobian is given in eq.~\eqref{eq:jacobian}. Moreover, we use eq.~\eqref{eqn:inverseJacobian} to obtain
\beq
\psi'_1(x) = h_1'(\tau)\,\partial_x\tau = h_1'(\tau)\,\frac{\mathcal{D}(\tau)}{h_1(\tau)^2}\,,
\eeq
and
\beq
\psi_2'(x) = \frac{\psi_2(x)\psi_1'(x)+D(x)}{\psi_1(x)} = \frac{\mathcal{D}(\tau)}{h_1(\tau)^2}\,\left(\tau\,{h_1'(\tau)}+{h_1(\tau)}\right)\,.
\eeq
Since $h_1$ is a modular form of weight 1, $h'_1$ is a quasi-modular form of weight 3 and depth at most 1, i.e., there are $A_1\in M_1(\Gamma_2)$ and $A_3\in M_3(\Gamma_2)$ such that $h_1' = A_1\,G_2 + A_3$. This gives:
\beq\bsp
dx\,&\psi_1(x)^{\alpha}\,\psi_2(x)^{\beta}\,\psi_1'(x)^{\gamma}\,\psi_2'(x)^{\delta} = \\
&\,=d\tau\,\mathcal{D}(\tau)^{\gamma+\delta-2}\,h_1(\tau)^{2+\alpha+\beta-2\gamma-2\delta}\,\tau^{\beta}\,(A_1\,G_2 + A_3)^{\gamma}\\
&\,\quad\times(\tau\,A_1(\tau)\,G_2(\tau)+\tau\,A_3(\tau)+h_1(\tau))^\delta\\
&\,=d\tau\,\mathcal{D}(\tau)^{\gamma+\delta-2}\,h_1(\tau)^{2+\alpha+\beta-2\gamma-2\delta}\,\tau^{\beta}\\
&\,\quad\times \sum_{p=0}^{\gamma}\sum_{q=0}^{\delta}\binom{\gamma}{p}\binom{\delta}{q}A_1(\tau)^{p+q}\,G_2(\tau)^{p+q}\,\tau^q\,A_3(\tau)^{\gamma-p}\,(\tau\,A_3(\tau)+h_1(\tau))^{\delta-q}\,.
\esp\eeq
The term with the highest power of $G_2$ is:
\beq
d\tau\,\mathcal{D}(\tau)^{\gamma+\delta-2}\,h_1(\tau)^{2+\alpha+\beta-2\gamma-2\delta}\,\tau^{\beta}\,A_1(\tau)^{\gamma+\delta}\,G_2(\tau)^{\gamma+\delta}\,\tau^\delta\,.
\eeq
It has depth $s=\gamma+\delta$ and weight $ k =\alpha+\beta+\gamma+\delta+2 = \alpha+\beta+s+2$. Since $\alpha$ and $\beta$ are positive integers, we have $k \ge s+2 > s+1$, and so we never reach the maximally allowed depth $s=k-1$. 

To finish the proof of Theorem~\ref{thm:section2}, we need to comment on the inhomogeneous term $N(x,\eps)$ in eq.~\eqref{eq:DEQ_generic_higher}. By assumption, the $\eps$-expansion of $N(x,\eps)$ involves at every order only sums of rational functions of $x$ multiplied by MPLs of the form $G(a_1,\ldots, a_n;x)$, with $a_i$ independent of $x$. It is easy to see that MPLs of this form can always be written as iterated integrals of modular forms for $\Gamma_2$, because
\beq
\frac{dx}{x-a_i} = \frac{h_1(\tau)^2\,d\tau}{\mathcal{D}(\tau)\,(x(\tau) - x(\tau_i))}\,,\textrm{~~with~~} x(\tau_i) = a_i\,.
\eeq
This finishes the proof of Theorem~\ref{thm:section2}.

\paragraph{Linear independence for iterated integrals.}
We have seen how Theorem~\ref{thm:main} leads to a proof of Theorem~\ref{thm:section2}, which characterises the iterated integrals that arise as solutions to a certain class of differential equations. In applications one is usually also interested in having a minimal set of of iterated integrals, i.e., a basis of linearly-independent iterated integrals. 
We now show how Theorem \ref{thm:main} yields the desired linear independence result. The crucial mathematical ingredient is a linear independence criterion for iterated integrals, which is very general and not at all limited to meromorphic modular forms. We first describe this criterion in a special case which, while being far from the most general possible statement, is sufficiently general to exhibit all essential features of the general case (for details, see ref.~\cite{DDMS}).

Suppose that $\mathcal{F}=\{f_i\}_{i\in I}$ is a family of meromorphic functions on the upper half-plane, and let $K$ be a differential subfield of the field of meromorphic functions on $\mathfrak{H}$, which contains all $f_i$. Here, `differential subfield' means that $K$ is a subfield which is closed under differentiation of meromorphic functions. The following theorem is a variant of a classical result due to Chen, \cite[Theorem 4.2]{ChenSymbol}.
\begin{thm}[{\cite[Theorem 2.1]{DDMS}}]
\label{thm:lindep}
The following assertions are equivalent.
\begin{itemize}
\item[(i)]
The family of all iterated integrals (viewed as functions of $\tau$) of the form
\[
\int^\tau_{i\infty}d\tau_1\,f_{i_1}(\tau_1)\int_{i\infty}^{\tau_1}d\tau_2\,f_{i_2}(\tau_2)\ldots \int_{i\infty}^{\tau_{n-1}}d\tau_{n}\,f_{i_n}(\tau_n) \, ,
\]
for all $n\geq 0$ and all $f_i \in \mathcal{F}$, is $K$-linearly independent.
\item[(ii)]
The family $\mathcal{F}$ is $\mathbb C$-linearly independent and we have 
\[
\partial_{\tau}(K)\cap \operatorname{Span}_\mathbb C\mathcal{F}=\{0\} \, ,
\]
where $\operatorname{Span}_\mathbb C\mathcal{F}$ denotes the vector space of all $\mathbb C$-linear combinations of functions in $\mathcal{F}$.
\end{itemize}
\end{thm}

Let us apply this theorem to our setting. Here $K:= \cM(\Gamma_2,R_S)(G_2)$ is the field of fractions of $\cQ\cM(\Gamma_2, R_S)$, i.e., the field whose elements are ratios of quasi-modular forms, or equivalently ratios of polynomials in $G_2$ with coefficients that are meromorphic modular forms. $K$ is a differential subfield, because quasi-modular forms are closed under differentiation. Clearly, we have $\cQ\cM(\Gamma_2, R_S)\subset K$. For $\mathcal{F}$ we choose
\beq
\mathcal{F} = \bigcup_{k\in\mathbb{Z}}\mathcal{F}_k\,,
\eeq
with $\mathcal{F}_k := \{f_1^{(k)},\ldots, f_{p_k}^{(k)}\}$ a $\mathbb{C}$-linearly independent set of modular forms from $\cM_k(\Gamma_2,R_S)$ for $k<2$ and from $\widetilde{\cM}_k(\Gamma_2,R_{s_0})$ for $k\ge 2$ and some fixed choice of cusp $s_0\in S_{\Gamma_2}$ (see section~\ref{sec:neat_proof} how to construct explicit bases for these vector spaces). Since the sums in Theorem~\ref{thm:main} are direct, it is easy to see that we have $\partial_{\tau}(K)\cap\textrm{Span}_{\mathbb{C}}\mathcal{F}=\{0\}$, and so Theorem~\ref{thm:lindep} implies that the corresponding iterated integrals are $K$-linear independent.

\subsection{Sketch of the proof for neat subgroups}
\label{sec:neat_proof}

We now return to the proof of Theorem~\ref{thm:main} for a special class of of subgroups.
\begin{defi}\label{defi:neat}
A subgroup $\Gamma\subseteq\slz$ is called \emph{neat} if $\left(\begin{smallmatrix}-1 & 0\\0&-1\end{smallmatrix}\right)\notin\Gamma$  and $\Gamma$ has no elliptic points nor irregular cusps, $\eps_2(\Gamma)=\eps_3(\Gamma)=\eps_i(\Gamma)=0$.
\end{defi}
One can show that
\begin{enumerate}
\item $\Gamma_1(N)$ is neat and has genus zero for $N\in\{5,\ldots,10,12\}$;
\item $\Gamma(N)$ is neat and has genus zero for $N\in\{3,4,5\}$.
\end{enumerate}
In particular, the congruence subgroup $\Gamma_1(6)$ relevant to the banana graph is neat and has genus zero. 

For $k<2$, the proof of Theorem~\ref{thm:main} is identical to the proof for $\Gamma=\slz$ considered in ref.~\cite{matthes2021iterated}, and we do not consider it here.
In order to see that Theorem~\ref{thm:main} holds also for $k\ge 2$, we first note that for every $f\in\cQ\cM_k(\Gamma,R_{s_0})$ there are meromorphic modular forms $h_1\in \cM_k(\Gamma,R_{s_0})$ and $h_2\in \cM_{2-k}(\Gamma,R_{s_0})$ and a meromorphic quasi-modular form $g\in \cQ\cM_{k-1}(\Gamma,R_{s_0})$ such that
$f = h_1 + h_2\,G_2^{k-1}+\delta g$.
This decomposition holds for all subgroups $\Gamma$, and does not require $\Gamma$ to be neat. It is a direct consequence of the algorithms of ref.~\cite{AMBP_2012__19_2_297_0}. Theorem~\ref{thm:main} then follows from the following claim: For every meromorphic modular form $f\in\cM_k(\Gamma,R_{s_0})$ of weight $k\ge2$ there is $h\in \widetilde{\cM}_{k}(\Gamma,R_{s_0})$ and $g\in\cQ\cM_{k-2}(\Gamma,R_{s_0})$ such that 
\beq\label{eq:decomp_1}
f=h+\delta g\,.
\eeq
In the remainder of this section we show how to construct $h$ and $g$ explicitly for neat subgroups of genus zero. Before doing so, we review some mathematical tools required to achieve this decomposition.

\paragraph{Bol's identity.} A complication when trying to decompose a meromorphic modular $f$ into an elements $h\in \widetilde{\cM}_{k}(\Gamma,R_{s_0})$ up to a total derivative is the fact that in general derivatives of modular forms are themselves not modular, but only quasi-modular. However, an important result due to Bol~\cite{Bol} states that we recover again a modular form if we take enough derivatives. More precisely, Bol's identity states that for $k\ge 2$ there is a linear map 
\beq\label{eq:Bol}
\delta^{k-1}: {\cM}_{2-k}(\Gamma) \to {\cS}_{k}(\Gamma)\,.
\eeq
In other words, if $k\ge 2$ and $f\in {\cM}_{2-k}(\Gamma)$, then in general $\delta f$ will not be a modular form, i.e., $\delta f\notin  {\cM}_{4-k}(\Gamma)$, but the $(k-1)^{\textrm{th}}$ derivative will be a modular form of weight $k$, $\delta^{k-1} f\in  {\cM}_{k}(\Gamma)$ (and in fact, it will even be a cusp form). Note that eq.~\eqref{eq:Bol} remains true if we replace ${\cM}_{2-k}(\Gamma)$ and ${\cS}_{k}(\Gamma)$ by ${\cM}_{2-k}(\Gamma,R_{s_0})$ and ${\cS}_{k}(\Gamma,R_{s_0})$ respectively. 

The main idea to achieve the decomposition in eq.~\eqref{eq:decomp_1} then goes as follows: Assume we are given $f\in \cM_{k}(\Gamma,R_{s_0})$ with a pole of order $m>0$ at a point $P\in R_{s_0}$, and assume the order of the pole is too high for $f$ to lie in $\widetilde{\cM}_{k}(\Gamma,R_{s_0})$. We will show how to construct $\tilde{g}\in\cM_{2-k}(\Gamma,R_{s_0})$ such that $f-\delta^{k-1}\tilde{g}$ has a pole of order at most $m-1$ at $P$. Applying this approach recursively, we obtain the decomposition in eq.~\eqref{eq:decomp_1}.


\paragraph{Divisors and the valence formula.} From the previous discussion it becomes clear that an important step in achieving the decomposition in eq.~\eqref{eq:decomp_1} is the construction of $\tilde{g}\in\cM_{2-k}(\Gamma,R_{s_0})$ with prescribed poles. An important tool to understand meromorphic functions (or, more generally, meromorphic sections of line bundles) on a Riemann surface are divisors, which we review in this section. The material in this section is well-known in the mathematics literature, but probably less so in the context of Feynman integrals.

A \emph{divisor} on $X_{\Gamma}$ is an element in the free group $\text{Div}(X_{\Gamma})$ generated by the points of $X_{\Gamma}$ (divided by their order $h_P$), i.e., a divisor is an expression of the form $D=\sum_{P\in X_{\Gamma}}\frac{n_P}{h_P} [P]$, where the $n_P$ are integers, and only finitely many of the $n_P$ are non zero.
We can use divisors to encode the information on the zeroes and poles of a meromorphic function or modular form. More precisely, if $0\neq f\in \cM_k(\Gamma)$, we can associate a divisor to it, defined by
\beq
(f) = \sum_{s\in S_\Gamma} \frac{\nu_s(f)}{h_s}\,[s] + \sum_{P\in X_{\Gamma}\setminus S_\Gamma} \frac{\nu_P(f)}{\#\overline{\Gamma}_P}\,[P]\,,
\eeq
where $h_s=2$ if $s$ is irregular and $h_s=1$ otherwise, and $\overline{\Gamma}_P$ is the projection of ${\Gamma}_P$ to $\pslz$ (i.e., we have identified $\gamma\in \slz$ and $-\gamma\in\slz$). Note that we have the obvious relation $(fg)=(f)+(g)$.

To every divisor $D=\sum_{P\in X_{\Gamma}}\frac{n_P}{h_P} [P]$ we can associate its \emph{degree} $\deg D = \sum_{p\in X_{\Gamma}}n_p$. Since every meromorphic function on a compact Riemann surface must have the same number of zeroes and poles (counted with multiplicity), we must have $\deg(f)=0$ for all $f\in \cM_0(\Gamma)$. 
For meromorphic modular forms of weight $k$, the degree of the associated divisor is no longer zero, but it is given by the \emph{valence formula}:
\beq\label{eq:valence_formula}
\deg(f) = k\,d_{\Gamma}\,.
\eeq


\paragraph{Modular forms for neat subgroups of genus zero.}
From now on we focus on the case where $\Gamma$ is neat and has genus zero. Equation~\eqref{eq:genus} then implies
\beq\label{eq:eps_infty_to_dGamma}
\eps_{\infty}(\Gamma) =2(1+ d_{\Gamma})\,.
\eeq
As we will now show, the spaces of meromorphic modular forms for neat subgroups can be described very explicitly. 
\begin{lemma}
Let $\Gamma$ be a neat subgroup of genus zero. Then there exists $\mH_k \in \cM_k(\Gamma)$ such that $\nu_\infty(\mH_k)=k\,d_{\Gamma}$ and $\nu_P(\mH_k)=0$ otherwise. In particular, for $k>0$, $\mH_k$ is a modular form of weight $k$ for $\Gamma$.
\end{lemma}
\begin{proof}
If $k=0$, we simply choose $\mH_0=1$, and for $k<0$ we set $\mH_k = 1/\mH_{-k}$. So it is sufficient to discuss $k>0$.
Since $\dim_{\mathbb{C}}\cM_k(\Gamma)\neq 0$,\footnote{This can be seen by thinking of (meromorphic) modular forms of weight $k$ for the group $\Gamma$ as (meromorphic) sections of a certain line bundle (the $k$-th power of the Hodge bundle) on the modular curve $X_\Gamma$. It then follows from the Riemann--Roch formula that every line bundle on a compact Riemann surface admits a meromorphic section. See also ref.~\cite{diamond2005first}, the discussion after Theorem 3.6.1, for a detailed proof.} it contains a meromorphic modular form $h$ with divisor
\beq
(h) = \sum_{P\in X_{\Gamma}}n_P\,[P] = kd_{\Gamma}\,[\infty] + D \,,
\eeq
 where we used the fact that $h_P=1$ for neat subgroups, and we defined
 \beq
 D := \left(n_{\infty}-kd_{\Gamma}\right)[\infty] + \sum_{\substack{P\in X_{\Gamma}\\ P\neq \infty}}n_P\,[P]\,.
 \eeq
 The valence formula implies $\deg D=0$, and so there is a meromorphic function $f\in\cM_0(\Gamma)$ such that $(f)=D$. Since $\Gamma$ has genus zero, every meromorphic function is a rational function in the Hauptmodul $\tz$, and it is sufficient to pick
 \beq
 f := \prod_{\substack{P\in X_{\Gamma}\\ P\neq \infty}}(\tz-P)^{n_P}\,.
 \eeq
 It is now easy to check that $\mH_k := \frac{1}{f}\,h$ has the desired property. The fact that $\mH_k$ is holomorphic follows from $\nu_P(\mH_k)\ge 0$ for all $P\in X_{\Gamma}$.
\end{proof}

Note that $\mH_k$ is unique, up to overall normalisation. Indeed, assume that $\mH_k^{(1)}$ and $\mH_k^{(2)}$ both satisfy the condition, then $(\mH_k^{(1)}/\mH_k^{(2)}) = (\mH_k^{(1)}) - (\mH_k^{(2)}) = 0$, and so there is $\alpha \in \mathbb{C}$ such that $\mH_k^{(1)} = \alpha \mH_k^{(2)}$. We assume from now on that the normalisation of $\mH_k$ is chosen such that at the infinite cusp we have the $q$-expansion ($h$ is the width at the infinite cusp):
\beq\label{eq:Hk_normalisation}
\mH_k(\tau) = q^{kd_{\Gamma}/h}\left[1 +\sum_{n\ge 1} a_n\,q^{n/h}\right]\,,\qquad q=e^{2\pi i\tau}\,.
\eeq

We can use $\mH_k$ to give an explicit representation of the spaces of modular forms of weight $k$ in terms of rational functions,
\beq\bsp
\cM_k(\Gamma) &\,= \mH_k\cdot \mathbb{C}(\tz)\,,
\esp\eeq
with holomorphic modular forms of weight $k$ corresponding to polynomials of degree at most $k\,d_{\Gamma}$:
\beq\bsp
M_k(\Gamma) &\,= \mH_k\cdot \mathbb{C}[\tz]_{\le kd_{\Gamma}}\,,
\esp\eeq
where $\mathbb{C}[X]_{\le m}$ denotes the vector space of polynomials of degree at most $m$. We can use this representation to write down a generating set for $\cM_k(\Gamma)$. For $p\in X_{\Gamma}$ and $m\in\mathbb{Z}_{>0}$, we define:
\beq\bsp
u_{P,m}(\tau) &\,= \left\{\begin{array}{ll}
(\tz(\tau)-P)^{-m}\,, & \text{ if }P\neq\infty\,,\\
\tz(\tau)^m\,,& \text{ if }P=\infty\,,
\end{array}\right.\\
u_{\infty,0}(\tau) &\,= 1\,.
\esp\eeq
It is an easy exercise (based on partial fractioning) to show that every rational function has a unique representation as a finite linear combination of the $u_{P,m}$. As a consequence, the meromorphic modular forms $U_{k,P,m} := \mH_k\,u_{P,m}$ are a generating set for $\cM_k(\Gamma)$. In particular, a basis for $M_k(\Gamma)$ is $\{\mH_k\,u_{\infty,m}:0\le m\le kd_{\Gamma}\}$. Moreover, we can use this generating set to write down an explicit basis for $\widetilde{\cM}_k(\Gamma,R_{\infty})$ in definition~\ref{defi:Mtilde}. For $k\ge 2$, we have
\beq\label{eq:Mtilde_def}
\widetilde{\cM}_k(\Gamma,R_{\infty}) := M_k(\Gamma) \cup \widehat{\cS}_k(\Gamma)\cup\widehat{\cM}_k(\Gamma,R)\,,
\eeq
where we defined
\beq\bsp\label{eq:Shat_def}
\widehat{\cM}_k(\Gamma,R):=\bigoplus_{\substack{p\in R\setminus \{\infty\} \\ 1\le m\le k-1}}\!\!\!\! \mathbb{C}\,U_{k,P,m} = \bigoplus_{\substack{P\in R \\ 1\le m\le k-1}}\!\!\!\! \mathbb{C}\,\frac{\mH_{k}}{(\tz-P)^m}\,,\\
\widehat{\cS}_k(\Gamma):=\bigoplus_{kd_{\Gamma}< m< 2d_{\Gamma}\,(k-1)}\!\!\!\!\!\!\!\!\!\! \mathbb{C}\,U_{k,\infty,m}=\bigoplus_{kd_{\Gamma}< m< 2d_{\Gamma}\,(k-1)} \!\!\!\!\!\!\!\!\mathbb{C}\,\mH_k\,\tz^m\,.
\esp\eeq
Note that $\dim_{\mathbb{C}}\widehat{\cS}_k(\Gamma) = \dim_{\mathbb{C}}S_k(\Gamma)$.

\paragraph{Sketch of the proof of Theorem~\ref{thm:main} for neat subgroups.}
We now show how we can use Bol's identity to construct for each $f=U_{k,P,m}$ a function $\tilde{g}$ such that the decomposition in eq.~\eqref{eq:decomp_1} holds. We will make repeated use of the following result:
\begin{lemma}\label{lem:valence} Let $\Gamma$ be neat and have genus zero, $f\in \cM_{2-k}(\Gamma)$, for $k\ge 2$.
\begin{enumerate}
\item If $P\in X_{\Gamma}\setminus S_{\Gamma}$, then $\nu_P(\delta^{k-1}f) \ge 0$ or $\nu_P(\delta^{k-1}f) =1-k+\nu_P(f)$.
\item If $s\in S_{\Gamma}$, then $\nu_s(\delta^{k-1}f) \ge 0$ or $\nu_s(\delta^{k-1}f)=\nu_s(f)$.
\end{enumerate}
\end{lemma}
\begin{proof}
It is sufficient to prove the claim for the elements of the generating set, $U_{2-k,P,m} = \mH_{2-k} u_{P,m} = \mH_{k-2}^{-1}u_{P,m}$. 

We need to show that if $\delta^{k-1}f$ has a pole at $P$, i.e., $\nu_P(\delta^{k-1}f)<0$, then it satisfies the claim of the lemma. Note that if $\nu_P(\delta^{k-1}f)<0$, then also $\nu_P(f)<0$.

Let $P\in X_{\Gamma}\setminus S_{\Gamma}$. We have $\nu_P(U_{2-k,P,m}) = -m<0$. Then, if $\tau_P$ is such that $t(\tau_P)=P$, $U_{k,P,m}$ admits a Laurent expansion of the form
\beq
U_{2-k,P,m}(\tau) = \frac{\alpha}{(\tau-\tau_P)^m} + \ord\left(\frac{1}{(\tau-\tau_P)^{m-1}}\right)\,,\qquad \alpha\in\mathbb{C}\setminus\{0\}
\eeq
and so
\beq\bsp
\delta^{k-1}U_{2-k,P,m}(\tau)&\, =(2\pi i)^{1-k}\partial_{\tau}^{k-1}U_{2-k,P,m}(\tau)\\
&\, = \frac{(m)_{k-2}\,\alpha}{(-2\pi i)^{k-1}\,(\tau-\tau_P)^{m+k-1}}+ \ord\left(\frac{1}{(\tau-\tau_P)^{m+k-2}}\right)\,,
\esp\eeq
where $(a)_n = a(a+1)\ldots(a+n)$.
Hence $\nu_P(\delta^{k-1}U_{2-k,P,m}) = 1-k-m = 1-k+\nu_P(U_{2-k,P,m})$.

Let $s\in S_{\Gamma}$, $s\neq [\infty]$. We have $\nu_P(U_{2-k,P,m}) = -m<0$. If $q$ is a local coordinate in a neighbourhood of the cusp $s$, $U_{2-k,P,m}$ admits the Fourier expansion:
\beq
U_{2-k,P,m}(q)=\frac{\alpha}{q^m} + \ord(q^{-m+1})\,,\qquad \alpha\in\mathbb{C}\setminus\{0\}\,,
\eeq
and so
\beq
\delta^{k-1}U_{2-k,P,m}(q) = (q\partial_q)^{k-1}U_{2-k,P,m}(q) = \frac{(-1)^{k-1}\,(m)_{k-2}\,\alpha}{q^m} + \ord(q^{-m+1})\,.
\eeq
Hence, $\nu_s(\delta^{k-1}U_{2-k,P,m}) = \nu_s(U_{2-k,P,m})$.

If $s=[\infty]$, then $\nu_{\infty}(U_{2-k,\infty,m}) = -m-\nu_{\infty}(\mH_{k-2}) = -m - (k-2)d_{\Gamma}$. By the same argument as in the case $s\neq[\infty]$, we conclude $\nu_{\infty}(\delta^{k-1}U_{2-k,P,m}) = -m - (k-2)d_{\Gamma} = \nu_{\infty}(U_{2-k,P,m})$.
\end{proof}

We are now in a position to prove the decomposition in eq.~\eqref{eq:decomp_1}. The proof is constructive, and allows one to recursively construct the functions $h$ and $\tilde{g}$ in eq.~\eqref{eq:decomp_1}. The decomposition in eq.~\eqref{eq:decomp_1} is equivalent to the following result:
\begin{thm}\label{thm:bijection} For $k\ge 2$, there is a decomposition 
\beq
\cM_k(\Gamma,R_S)=\widetilde{\cM}_k(\Gamma,R_\infty)\oplus \delta^{k-1}\cM_{2-k}(\Gamma,R_S)\,.
\eeq
\end{thm}
\begin{proof}
It is sufficient to consider the case $\#R=1$.

We first show surjectivity. For this it is sufficient to show that all those $U_{k,P,m}$ not in $\widetilde{\cM}_k(\Gamma,R_{\infty})$ do not define independent classes modulo objects that lie in the image of Bol's identity, i.e., these classes can be expressed as linear combinations in $\widetilde{\cM}_k(\Gamma,R_{\infty})$, modulo $\delta^{k-1}\cM_{2-k}(\Gamma,R_S)$. 

Let $s\in S_{\Gamma}$, $s\neq [\infty]$ and $m>0$. Let $q$ be a local coordinate around $s$. Then there are $\alpha_1,\alpha_2\in\mathbb{C}\setminus\{0\}$ such that
\beq\bsp
U_{k,s,m}(q) &\,= \frac{\alpha_1}{q^m} + \ord(q^{-m+1})\,,\\
U_{2-k,s,m}(q) &\,= \frac{\alpha_2}{q^m} + \ord(q^{-m+1})\,.
\esp\eeq
Lemma~\ref{lem:valence} implies 
\beq
U_{k,s,m}(q) - \frac{\alpha_1}{\alpha_2}\,(-m)^{1-k}\,\delta^{k-1}U_{2-k,s,m}(q) = \ord(q^{-m+1})\,.
\eeq
Applying this identity recursively, we arrive at the conclusion that 
\beq
U_{k,s,m} = 0\!\!\! \mod \delta^{k-1}\cM_{2-k}(\Gamma,R_S)\,,\quad \text{for all } m>0\,.
\eeq

For the infinite cusp, we know from the proof of Lemma~\ref{lem:valence} that $\nu_\infty(\delta^{k-1}U_{2-k,\infty,m'}) = \nu_\infty(U_{2-k,\infty,m'}) = -m'-(k-2)d_{\Gamma}$, for all $m'>0$. 
Hence, for all $m\ge 2d_{\Gamma}(k-1)$ there is $m' = m-2d_{\Gamma}(k-1)\ge 0$, and we can pick a local coordinate $q$ at the infinite cusp such that there is $\alpha_1,\alpha_2\neq0$ such that
\beq\bsp
U_{k,\infty,m}(q) &\,= \frac{\alpha_1}{q^{m-kd_{\Gamma}}} + \ord(q^{-m+kd_{\Gamma}+1})\,,\\
U_{2-k,\infty,m'}(q) &\,= \frac{\alpha_2}{q^{m'-(2-k)d_{\Gamma}}} + \ord(q^{-m'+(2-k)d_{\Gamma}+1})= \frac{\alpha_2}{q^{m-kd_{\Gamma}}} + \ord(q^{-m+kd_{\Gamma}+1})\,.
\esp\eeq
Lemma~\ref{lem:valence} then implies
\beq
U_{k,\infty,m}(q) - \frac{\alpha_1}{\alpha_2}\,(kd_{\Gamma}-m)^{1-k}\,\delta^{k-1}U_{2-k,\infty,m'}(q) = \ord(q^{-m+kd_{\Gamma}+1})\,.
\eeq
It follows that $U_{k,\infty,m}$ for $m\ge 2d_{\Gamma}(k-1)$ does not define an independent class modulo total derivatives. 

Finally, let $R=\{P\}$, and $m>k-1$. We can take $m'=m-k+1>0$, and Lemma~\ref{lem:valence} implies $\nu_P(\delta^{k-1}U_{2-k,P,m'}) = 1-k-m'=-m$. Hence, with $\tau_P$ such that $t(\tau_P)=P$, there is $\alpha_1,\alpha_2\neq0$ such that
\beq\bsp
U_{k,P,m}(\tau) &\,= \frac{\alpha_1}{(\tau-\tau_P)^m} + \ord((\tau-\tau_P)^{-m+1})\,,\\
\delta^{k-1}U_{2-k,P,m}(\tau) &\,= \frac{\alpha_2}{(\tau-\tau_P)^{m'+k-1}} + \ord((\tau-\tau_P)^{-m'-k+2})\\
&\, = \frac{\alpha_2}{(\tau-\tau_P)^m} + \ord((\tau-\tau_P)^{-m+1})\,.
\esp\eeq
Hence
\beq
U_{k,P,m}(\tau)-\frac{\alpha_1}{\alpha_2}\delta^{k-1}U_{2-k,P,m}(\tau)  = \ord((\tau-\tau_P)^{-m+1})\,,
\eeq
and so $U_{k,P,m}$ for $m>k-1$ does not define an independent class modulo total derivatives. This finishes the proof of surjectivity.

Let us now show injectivity (for $R=\{P\}$). Consider the following general linear combination of elements from $\widetilde{\cM}_k(\Gamma,R_{\infty})$:
\beq\bsp
f&\,:=\sum_{m=1}^{k-1}\alpha_m\,U_{k,P,m} + \sum_{n=0}^{2d_{\Gamma}(k-1)-1}\beta_n\,U_{k,\infty,n}\\
&\,=\mH_k\sum_{m=1}^{k-1}\frac{\alpha_m}{(\tz-P)^m}+ \mH_k\sum_{n=0}^{2d_{\Gamma}(k-1)-1}\beta_n\tz^n\,.
\esp\eeq
We need to show that whenever there is $g\in \cM_{2-k}(\Gamma,R_S)$ such that $f=\delta^{k-1}g$, then necessarily $\alpha_m=0$ for $1\le m<k$ and $\beta_m=0$ for $0\le n<2d_{\Gamma}(k-1)$. Let us start by showing that the coefficients $\alpha_m$ must vanish. To see this, assume that $\alpha_m\neq 0$ for some $1\le m<k$. Then $f$ has a pole at $\tz=P$, i.e., $0>\nu_P(f)>-k$. Hence, $g$ must also have a pole at $\tz=P$, i.e., $\nu_P(g)<0$. Lemma~\ref{lem:valence} then implies $\nu_P(f) = \nu_P(\delta^{k-1}g) = 1-k+\nu_P(g)\le -k$, which is a contradiction. Hence $\alpha_m=0$ for all $1\le m<k$.

Next, let us assume that $\beta_n\neq0$ for some $2d_{\Gamma}k<n<2d_{\Gamma}(k-1)$. Then $f$ has a pole at the infinite cusp, with $0>\nu_{\infty}(f) \ge \nu_{\infty}(U_{k,\infty,2d_{\Gamma}(k-1)-1}) = (2-k)d_{\Gamma}+1$. Then $g$ must also have a pole. The order of the pole is bounded by
\beq
\nu_{\infty}(f) = \nu_{\infty}(\delta^{k-1}g) \le \nu_\infty(\delta^{k-1}U_{2-k,\infty,0}) = (2-k)d_{\Gamma}< \nu_{\infty}(f)\,,
\eeq
which is a contradiction. Hence $\beta_n=0$ for all $2d_{\Gamma}k<n<2d_{\Gamma}(k-1)$. It follows that $f$ must be a holomorphic modular form of weight $k$, but $M_k(\Gamma)\cap \delta^{k-1}(\cM_{2-k}(\Gamma)) = 0$ for $k\ge0$, and therefore $f=0$.
\end{proof}

\vspace{1mm}\noindent




\section{The monodromy groups of the equal-mass sunrise and banana integrals}
\label{sec:sunban}

In the remainder of this paper we will illustrate the abstract mathematical concepts on two very concrete families of Feynman integrals, namely the equal-mass two-loop sunrise and three-loop banana integrals, defined by:
\begin{align}
  \label{eq:sunrise-family} I^{\sun}&_{a_1,\dots,a_5}(p^2,m^2;d)=\\
  \nonumber& = \int \prod_{i=1}^2 \mathfrak{D}^d \ell_i 
		 \frac{(\ell_1\cdot p)^{a_4}(\ell_2\cdot p)^{a_5}}{[\ell_1^2-m^2]^{a_1}[\ell_2^2-m^2]^{a_2}[(\ell_1-\ell_2-p)^2-m^2]^{a_3}}\,,\\
  \label{eq:banana-family} I^{\ban}&_{a_1,\dots,a_9}(p^2,m^2;d)=\\
  \nonumber& = \int \prod_{i=1}^3 \mathfrak{D}^d \ell_i 
		 \frac{(\ell_3^2)^{a_5}(\ell_1\cdot p)^{a_6}(\ell_2\cdot p)^{a_7}(\ell_3\cdot
  p)^{a_8}(\ell_1\cdot\ell_2)^{a_9}}{[\ell_1^2-m^2]^{a_1}[\ell_2^2-m^2]^{a_2}[(\ell_1-\ell_3)^2-m^2]^{a_3}[(\ell_2-\ell_3-p)^2-m^2]^{a_4}}\,,
\end{align}
where $a_i\ge 0$ are positive integers, $m^2>0$ and $p^2$ are real. We work in dimensional regularisation in $d=2-2\eps$ dimensions. The integration measure reads 
\begin{equation}
  \label{eqn:intemeasure}
  \int\ddl=\frac{1}{\Gamma\left( 2-\frac{d}{2} \right)}\int\frac{d^d\ell}{i\pi^{d/2}}\,.
\end{equation}
We follow refs.~\cite{Remiddi:2016gno,Primo:2017ipr} for the choice of master integrals and the differential equations (see also section~\ref{sec:bananameromorphic} and appendix~\ref{app:sunban}).

In this section we focus on identifying the maximal cuts of these integrals as modular forms for the congruence subgroup $\Gamma_1(6)$, and in the next section we see how iterated integrals of meromorphic modular forms arise. This gives another way to resolve the debate in the literature whether the two-loop sunrise integral is associated with modular forms for $\Gamma_1(12)$ or $\Gamma_1(6)$; see, e.g., refs.~\cite{Adams:2017ejb,Frellesvig:2021vdl}. While the discussion in this section focuses on these specific Feynman integrals, it is easy to transpose the discussion to other differential operators of degree two or three. This may then provide a roadmap to identify the modular forms obtained from maximal cuts of one-parameter families of Feynman integrals that are not described by the same Picard-Fuchs operators as the examples considered here.

\subsection{The sunrise family}
\label{ssec:sunrisefamily}

\paragraph{The monodromy group.}
The maximal cuts of the integral $I^{\sun}_{1,1,1,0,0}(p^2,m^2;2)$ are annihilated by the second-order differential operator~\cite{Laporta:2004rb}:
\begin{align}
	\label{eq:sunriseDO}
	\cL^\sun_t&:=\partial_t^2 + \left(\frac{1}{t-9}+\frac{1}{t-1}+\frac{1}{t}\right)\partial_t +
	\left(\frac{1}{12(t-9)}+\frac{1}{4(t-1)}-\frac{1}{3t}\right)\,,
\end{align}
where we defined
\beq\label{eq:tt_def}
t := \frac{p^2}{m^2}\,.
\eeq
It is well-known that $\cL^\sun_t$ is the Picard-Fuchs operator describing a family of elliptic curves. Consequently, its solutions can be expressed in terms of elliptic integrals of the first kind (see appendix~\ref{app:sun}). In the following we explicitly construct the solutions using the Frobenius method reviewed in section~\ref{sec:frobenius_review} in order to outline the general strategy. While we only perform the calculations for the differential operator $\cL^\sun_t$, the different steps can be applied very generally to second-order differential operators describing one-parameter families of elliptic curves. 

\begin{figure}[h]
	\begin{center}\includegraphics{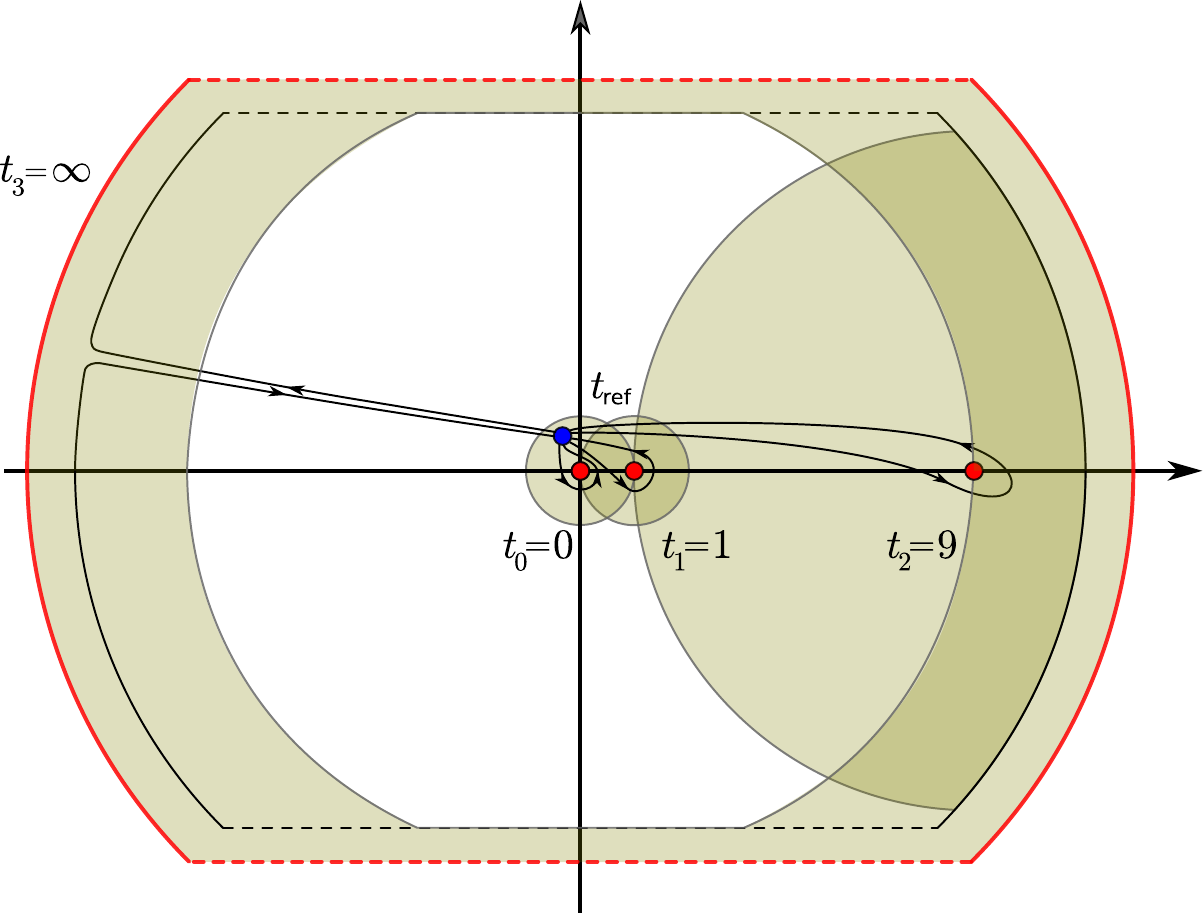}\end{center}
	\caption{Geometry associated to the sunrise differential operator $\cL^\sun_t$ in eq.~\eqref{eq:sunriseDO}. The coefficient functions have poles at $(t_0,\ldots,t_3)=(0,1,9,\infty)$. The corresponding radii of convergence are shaded in green.}
	\label{fig:sunmonodromy}
\end{figure}

The coefficients in $\cL^\sun_t$ have poles at $(t_0,\ldots,t_3)=(0,1,9,\infty)$, all of which are regular singular points. One can show that all of these points are MUM-points~\cite{Bonisch:2020qmm}, and close to each singular point we can choose a basis of solutions that consists of one holomorphic and one logarithmically-divergent function. For the singular point $t_0$ at the origin, the Frobenius method delivers the two power series solutions whose first terms in the expansion read:
\beq\bsp
\label{eqn:sunzerolocsolall}
	\phi_2(t_0;t)&=\frac{4 t}{9} + \frac{26t^2}{81}  + \frac{526t^3}{2187}  + \frac{1253t^4}{6561}  +\mathcal{O}(t^5) + \phi_1(t_0;t)\log t\,,\\
	\phi_1(t_0;t)&=1 + \frac{t}{3} + \frac{5t^2}{27} + \frac{31t^3}{243}  + \frac{71t^4}{729}   +\mathcal{O}(t^5)\,.
\esp\eeq
It is easy to check that these two local solution can be used to express the functions $\Psi_1(t)$ and $\Psi_2(t)$ in eq.~\eqref{eq:psi1_def}:
\beq\bsp\label{eq:phi_to_psi}
\Psi_1(t)&\, = \frac{2\pi}{\sqrt{3}}\,\phi_1(0;t)\,,\\
\Psi_2(t)&\, =-\frac{i}{\sqrt{3}}\,\phi_2(0;t)-\frac{2i}{\sqrt{3}}\,\log(3)\,\phi_1(0;t)\,.
\esp\eeq

Given the form of the solutions in eq.~\eqref{eqn:sunzerolocsolall}, it is not difficult to derive a representation of the local monodromy: the logarithm will acquire a phase of $2\pi i$ when transported around the pole at $t=t_0=0$ counterclockwise. Accordingly, one finds
\begin{align}
	\label{eqn:locmonzeropre}
\begin{pmatrix}\phi_2(t_0,t) \\ \phi_1(t_0;t)\end{pmatrix}_{\circlearrowleft}=\begin{pmatrix}1 & -2\pi i \\ 0 & 1 \end{pmatrix}\begin{pmatrix}\phi_2(t_0;t) \\ \phi_1(t_0;t)\end{pmatrix}. 
\end{align}
Repeating the calculation for $t_1=1$ and $t_2=9$ shows that the structure of the solutions equals those in eq.~\eqref{eqn:sunzerolocsolall}, just the coefficients are different. This is expected, since all singular points are MUM-points. Accordingly, the three local monodromy matrices are equal: 
\begin{align}
	\label{eqn:locmonzeron}
	\rho_{0}(\gamma_{0})=\rho_{1}(\gamma_{1})=\rho_{9}(\gamma_{9})=\rho_{\infty}(\gamma_{\infty})=\begin{pmatrix}1 & -2\pi i \\ 0 & 1 \end{pmatrix}. 
\end{align}
Matching the local solutions according to eq.~\eqref{eqn:calcmatchingmatrices} leads to:
\begin{subequations}
\begin{align}
	R_{0,1}&=
               \left(
               \begin{array}{cc}
                -\frac{3 \sqrt{3} \log (3)}{2 \pi } & \frac{24 \sqrt{2} \log (3)-\sqrt{3} \pi ^3}{2 \pi ^2}+\frac{3}{2} i \sqrt{3} \log (3) \\
                -\frac{3 \sqrt{3}}{4 \pi } & \frac{6 \sqrt{2}}{\pi ^2}+\frac{3 i \sqrt{3}}{4} \\
               \end{array}
               \right),\\
	R_{1,9}&=
		\left(
		\begin{array}{cc}
		 -\frac{8 \sqrt{2}}{3 \pi ^2} & \frac{1}{3} \sqrt{\frac{2}{3}} \pi  \log ^2(3)+\frac{8 i \sqrt{2}}{3 \pi } \\
		 -\frac{1}{\sqrt{3} \pi } & \frac{8 i \sqrt{6}+\pi ^2+\sqrt{2} \pi ^2 \log ^2(3)}{24 \sqrt{2}} \\
		\end{array}
		\right)\,,\\
	R_{9,\infty}&=
		\left(
		\begin{array}{cc}
		 -\frac{3 \pi^2 +3\sqrt{2} \pi^2  \log ^2(3)}{4 \sqrt{2}} & -4 \sqrt{3} \pi  \\
		 -\frac{6 \sqrt{3}}{\pi } & 0 \\
		\end{array}
		\right)
	\,.
\end{align}
\end{subequations}
Using eq.~\eqref{eqn:monodromytranslation} one can straightforwardly calculate the monodromy matrices in the basis of solutions $\phi(0;t)$:
\begin{align}
\rho_0(\gamma_0)&=\left(
\begin{array}{cc}
 1 & -2 i \pi  \\
 0 & 1 \\
\end{array}
\right)\,,\qquad
\rho_0(\gamma_9)=\left(
\begin{array}{cc}
 1-\frac{6 i \log (3)}{\pi } & \frac{12 i \log ^2(3)}{\pi } \\
 -\frac{3 i}{\pi } & 1+\frac{6 i \log (3)}{\pi } \\
\end{array}
\right),\\
\nonumber
\rho_0(\gamma_1)&=\left(
\begin{array}{cc}
 7-\frac{18 i \log (3)}{\pi } & -\frac{4 i (\pi -3 i \log (3))^2}{\pi } \\
 -\frac{9 i}{\pi } & -5+\frac{18 i \log (3)}{\pi } \\
\end{array}
\right),
\rho_0(\gamma_\infty)=\left(
\begin{array}{cc}
 7-\frac{12 i \log (3)}{\pi } & -\frac{6 i (\pi -2 i \log (3))^2}{\pi } \\
 -\frac{6 i}{\pi } & -5+\frac{12 i \log (3)}{\pi } \\
\end{array}
\right)\,.
\end{align}
This form of the monodromy matrices is not very enlightening. However, we can choose a basis of solutions such that the entries of the monodromy matrices have integer entries. Using a little algebra one can show that conjugation with the following matrix will bring all three monodromy matrices into integral form: 
\begin{align}
\label{eq:sunrise_a_p_rot}
a\,\left(
\begin{array}{cc}
	1 & -2 \log (3)+2{\pi i}\,p \\
	0 & 2 i \pi\,s  \\
\end{array}
\right)\,,\qquad p\in\mathbb{Z},\,s=\pm 1.
\end{align}
While the scaling parameter $a$ will drop out in the similarity transformation, choosing different parameters $p$ and $s$ will lead to different choices of generators. Setting $p=0$ and $s=-1$ for example leads to the following four matrices:
\beq\bsp
\label{eqn:sunresult}
\tilde{\rho}_0(\gamma_0)=\left(
\begin{array}{cc}
 1 & 1 \\
 0 & 1 \\
\end{array}
\right)\!,&\qquad
\tilde{\rho}_0(\gamma_1)=\left(
\begin{array}{cc}
 1 & 0 \\
 -6 & 1 \\
\end{array}
\right)\!,\\
\tilde{\rho}_0(\gamma_9)=\left(
\begin{array}{cc}
 7 & 2 \\
 -18 & -5 \\
\end{array}
\right)\!,
&\qquad \tilde{\rho}_0(\gamma_\infty)=\left(
\begin{array}{cc}
 7 & 3 \\
 -12 & -5 \\
\end{array}
\right)\,.
\esp\eeq
Combining the four matrices according to the succession of poles in figure~\ref{fig:sunmonodromy}, one finds $\tilde{\rho}_0(\gamma_{\infty})\tilde{\rho}_0(\gamma_9)\tilde{\rho}_0(\gamma_1)\tilde{\rho}_0(\gamma_0)=(\begin{smallmatrix}1 & 0\\0 & 1\end{smallmatrix}$). 
Note that the change of basis in eq.~\eqref{eq:phi_to_psi} corresponds to $(a,s,p) = (-i/\sqrt{3},1,0)$. This shows that the monodromy matrices in the basis $(\Psi_2(t),\Psi_1(t))$ have integer entries.

The matrices $\tilde{\rho}_0(\gamma_{i})$, $i\in\{0,1,9\}$, generate the monodromy group of $\cL_2^{\mathsf{sun}}$. We can see that 
\beq
\tilde{\rho}_0(\gamma_{i}) = \left(\begin{matrix}1 & \ast \\ 0 &1\end{matrix}\right)\!\!\!\! \mod 6\,.
\eeq
Since this relation holds for the generators, it must hold for all elements of the monodromy group, and so we see that the monodromy group must be a subgroup of $\Gamma_1(6)$ (cf.~eq.~\eqref{eq:Gamam1(N)_def}). We can show that the converse is also true. A short crosscheck shows that the generators for $\Gamma_1(6)$ delivered\footnote{The SAGE command to get a $2{\times}2$ matrix representation of a minimal set of generators of $\Gamma_1(6)$ reads \texttt{Gamma1(6).generators()} and delivers $(\begin{smallmatrix}1 & 1 \\ 0 & 1\end{smallmatrix}),\,(\begin{smallmatrix}-5 & 1 \\ -6 & 1\end{smallmatrix})$ and $(\begin{smallmatrix}7 & -3 \\ 12 & -5\end{smallmatrix})$.} by SAGE~\cite{SAGE} indeed generate the matrices in eq.~\eqref{eqn:sunresult}. Checking furthermore independence of the three matrices and noting that $\dim\Gamma_1(6)=3$, the monodromy group of $\cL_2^{\mathsf{sun}}$ is indeed $\Gamma_1(6)$.

\paragraph{The modular forms for the sunrise graph.}
Having identified the monodromy group of the two-loop sunrise integral with the congruence subgroup $\Gamma_1(6)$, it follows from the general discussion in section~\ref{sec:DEQs} that we expect the maximal cuts to define modular forms of weight 1 for $\Gamma_1(6)$. This agrees with the findings of refs.~\cite{Bloch:2013tra,Adams:2017ejb}. To make this explicit, we introduce a modular parametrisation and we
define the new variables $\tau$ and $q$ by (cf. eq.~\eqref{eq:tau_def_generic}):
\beq\bsp\label{eq:tau_def}
\tau = \frac{\Psi_2(t)}{\Psi_1(t)} &\,= \log(t/9) + \frac{4t}{9} + \frac{14t^2}{81} + \ord(t^3)\,,\\
q &\, =e^{2\pi i\tau} = \frac{t}{9} + \frac{4t^2}{81} + \ord(t^3)\,.
\esp\eeq
Note that we have chosen $\Psi_1(t)$ and $\Psi_2(t)$ such that $\Im\tau>0$ for $t\in X^{\mathsf{sun}}=\mathbb{P}^1_{\mathbb{C}}\setminus\{0,1,9,\infty\}$, and so $\tau\in \HP$. The change of variable from $t$ to $q$ is holomorphic at $t=0$ and can be inverted to express $t$ in terms of $q$. The result is the well-known expression for Hauptmodul for $\Gamma_1(6)$ in terms of Dedekind's eta function~\cite{Maier,Bloch:2013tra}:
 \beq\label{eq:t_def}
 t(\tau) = 9\,\frac{\eta(\tau)^4\eta(6\tau)^8}{\eta(2\tau)^8\eta(3\tau)^4} = 9\,q+\ord(q^2)\,,
 \eeq
 where $\eta(\tau)$ denotes Dedekind's eta function:
 \beq
 \eta(\tau) = e^{i\pi\tau/12}\,\prod_{n=1}^{\infty}(1-e^{2\pi in\tau})\,.
 \eeq
 It is easy to check (e.g., by comparing $q$-expansions with the basis of modular forms given by {\tt Sage}) that the function $h_1(\tau) := \Psi_1(t(\tau))$ defines a modular form of weight 1 for $\Gamma_1(6)$, as expected. In fact, it admits itself an expression in terms of Dedekind eta functions~\cite{Maier,Bloch:2013tra,Adams:2017ejb}:
 \beq
 h_1(\tau) = \frac{2\pi}{\sqrt{3}}\,\frac{\eta(2\tau)^6\eta(3\tau)}{\eta(\tau)^3\eta(6\tau)^2}\,.
 \eeq
Note that this is another way to see that the congruence subgroup naturally associated to the two-loop equal-mass sunrise integrals is $\Gamma_1(6)$ rather than $\Gamma_1(12)$, in agreement with the analysis in the literature~\cite{Adams:2017ejb,Frellesvig:2021vdl}. Moreover, we emphasise that nothing in our analysis is specific to the sunrise integral, and the exact same reasoning can be applied to other second-order differential operators describing families of elliptic curves, in particular those that appear in Feynman integrals computations.

\subsection{The banana family}
\paragraph{The third-order operator for the banana graph.}
We now repeat the computations of the previous section in the case of the three-loop equal-mass banana integral. Our goal is to determine the monodromy group of the differential operator in eq.~\eqref{eq:L_ban_3}. The 
calculation is slightly more complicated than in the two-loop case because the differential operator that annihilates the maximal cuts of $I^{\ban}_{1,1,1,1,0,0,0,0,0}(p^2,m^2;2)$ is of order three~\cite{Bloch:2014qca,Bloch:2016izu,Primo:2017ipr}:
\begin{equation}\label{eq:L_ban_3}
	\cL^{\ban,(3)}_x = \partial_x^3+\frac{3(8x-5)}{2(x-1)(4x-1)}\partial_x^2+\frac{4x^2-2x+1}{(x-1)(4x-1)x^2}\partial_x+\frac{1}{x^3(4x-1)}\,,
\end{equation}
with 
\beq
x = \frac{4m^2}{p^2}\,.
\eeq
In general, finding the kernel of a high-order operator can be a monumental task, and no closed form for the solution is necessarily known. 
The kernel of $\cL^{\ban,(3)}_x$ can be determined by noting that it is~\cite{Primo:2017ipr,joyce} the symmetric square
of (see subsection~\ref{ssec:ClassModularParametrization})
\begin{equation}\label{eq:L2x}
 \cL^{\ban,(2)}_x = \partial_x^2 + \frac{8x-5}{2(x-1)(4x-1)}\partial_x-\frac{2x-1}{4x^2(x-1)(4x-1)}\,.
\end{equation} 
The fact that $\cL^{\ban,(3)}_x$ is a symmetric square has a geometric origin: The $l$-loop equal-mass banana integral is associated to a one-parameter family of Calabi-Yau $(l-1)$-folds~\cite{Bloch:2014qca,Bloch:2016izu,Klemm:2019dbm,Bonisch:2020qmm,Bonisch:2021yfw}, and the maximal cuts of the $l$-loop equal-mass banana integral are annihilated by the Picard-Fuchs operator for this family, which has degree $l$. It is expected that the degree-three Picard-Fuchs operator of a one-parameter family of Calabi-Yau two-folds (also known as K3 surfaces) is always the symmetric square of a Picard-Fuchs operator describing a one-parameter family of elliptic curves, cf., e.g., ref.~\cite{Doran:1998hm}.

If we want to apply the results of section~\ref{sec:DEQs}, in particular Theorem~\ref{thm:section2}, we need all singular points of the second-order operator to be MUM-points. This, however, is not the case here, but only the singularities at $x=0$ and $x=\infty$ are MUM-points. 
We therefore perform the change of variables 
\begin{equation}\label{eq:change_of_vars}
  x(t)=\frac{-4\,t}{(t-1)(t-9)}\,.
\end{equation}
After this change of variables, one can see that $\textrm{Sol}(\cL^{\ban,(2)}_x)$ is spanned by $\sqrt{t}\Psi_1(t)$ and $\sqrt{t}\Psi_2(t)$, with $(\Psi_2(t),\Psi_1(t))$ defined in eq.~\eqref{eq:phi_to_psi}, i.e., they form a basis for $\textrm{Sol}(\cL_t^{\sun})$. It follows that 
\beq
\textrm{Sol}(\cL^{\ban,(3)}_x) = \mathbb{C}\,t\Psi_1(t)^2 \oplus  \mathbb{C}\,t\Psi_1(t)\Psi_2(t)  \oplus\mathbb{C}\,t\Psi_2(t)^2\,.
\eeq
In other words, we see that the solution space has the structure of the solution space of a symmetric square (up to the overall factor of $t$). The change of variables in eq.~\eqref{eq:change_of_vars} is 2-to-1, and the four MUM-points $t\in \{0,1,9,\infty\}$ of $\cL_x^{\sun}$ are mapped to the two MUM-points   $x\in \{0,\infty\}$ of $\cL^{\ban,(2)}_x$. The upshot is that after the change of variables in eq.~\eqref{eq:phi_to_psi}, Theorem~\ref{thm:section2} applies, and we expect the three-loop equal-mass banana integrals to be expressible in terms of iterated integrals of meromorphic modular forms for $\Gamma_1(6)$. We will investigate this in detail in section~\ref{sec:bananameromorphic}. In the remainder of this section we analyse the monodromy group associated to the banana integral in more detail.

\paragraph{The monodromy group.}
There are four regular singular points, and
the coefficient functions in $\cL_x^{\mathsf{ban},(3)}$ have poles at $(x_0,x_1,x_2,x_3)=(0,1/4,1,\infty)$ (see figure~\ref{fig:banana}).
\begin{figure}
\begin{center}\includegraphics{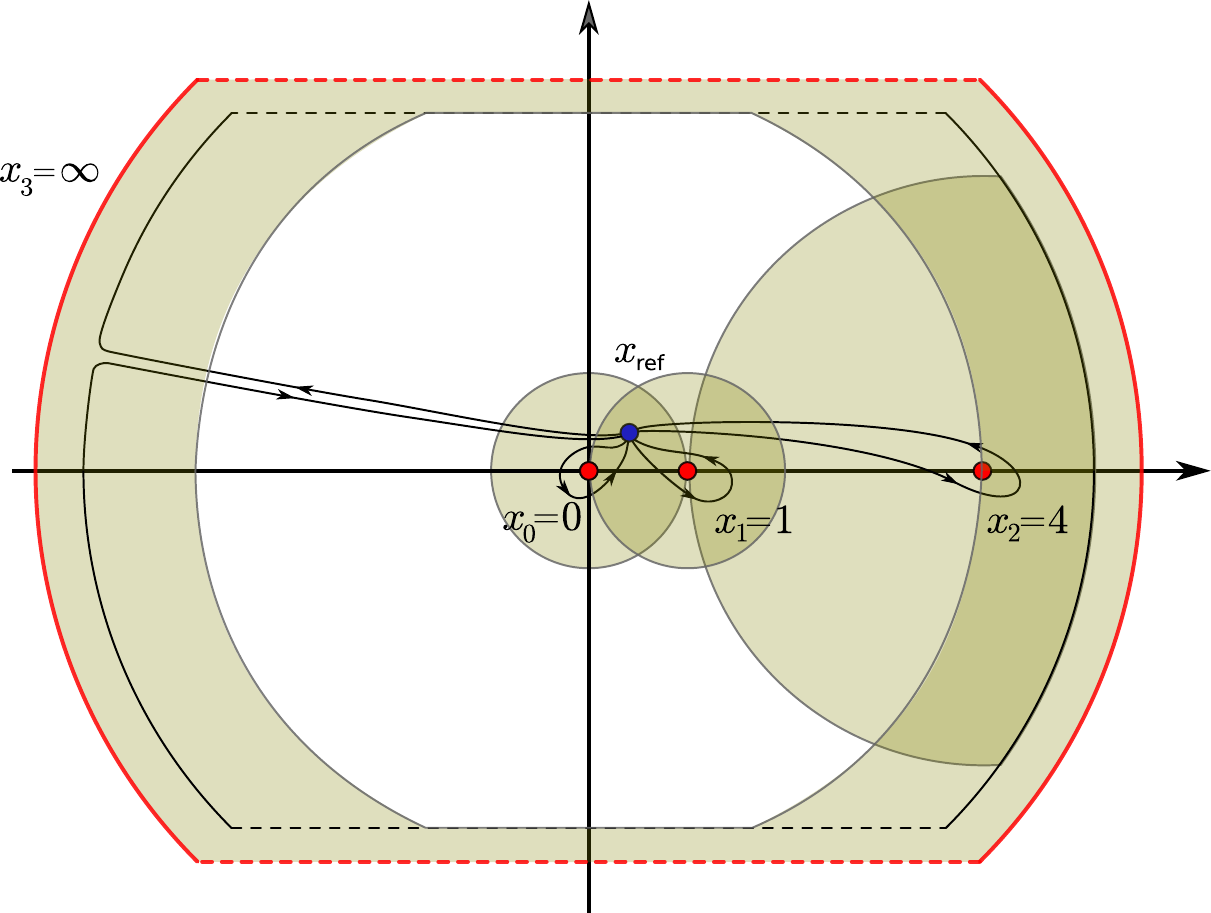}\end{center}
\caption{Geometry of the banana differential operator in eq.~\eqref{eq:L_ban_3}. The coefficient functions have poles at $(x_0,x_1,x_2,x_3)=(0,1/4,1,\infty)$. The corresponding radii of convergence are shaded in green.} 
\label{fig:banana}
\end{figure}
The singular point $x_0=0$ is a MUM-point, and the Frobenius method delivers three solutions, which read:
\beq\bsp
	\label{eqn:banzerolocsolall}
	\chi_3(x_0;x)&=x \left(\frac{9 x^2}{4}+\frac{135 x^3}{16}+\frac{7089 x^4}{256}+ \mathcal{O}(x^5)\right)+2 \chi_2(x_0;x) \log (x)\\
	&\,+\chi_1(x_0;x) \log ^2(x)\,,\\
\chi_2(x_0;x)&=x \left(\frac{3 x}{2}+\frac{57 x^2}{16}+\frac{73 x^3}{8}+\frac{13081 x^4}{512} + \mathcal{O}(x^5)\right)+\chi_1(x_0;x) \log (x)\,,\\
\chi_1(x_0;x)&=x \left(1 + x + \frac{7 x^2}{4} + 4 x^3 + \frac{679 x^4}{64} + \mathcal{O}(x^5)\right)\,.
\esp\eeq
This basis is related to the basis of ref.~\cite{Primo:2017ipr} (see also eq.~\eqref{eq:H1_to_Psi1}) via a constant rotation
\begin{align}
	\label{eqn:Brot}
	\begin{pmatrix}\chi_3(x_0;x)\\\chi_2(x_0;x)\\\chi_1(x_0;x)\end{pmatrix} = B
	\begin{pmatrix}I_1(x)\\J_1(x)\\H_1(x)\end{pmatrix},
\end{align}
with
\begin{align}
B=\left(
\begin{array}{ccc}
 \frac{4}{3} & -\frac{8 \log (2)}{\pi }+\frac{4 i}{3} & -1+\frac{4 \log ^2(2)}{\pi ^2}-\frac{4 i \log (2)}{\pi } \\
 0 & -\frac{2}{\pi } & \frac{2 \log (2)}{\pi ^2}-\frac{i}{\pi } \\
 0 & 0 & \frac{1}{\pi ^2} \\
\end{array}
\right)\,.
\end{align}
In eq.~\eqref{eq:H1_to_Psi1} we also show how the functions $(H_1(x), I_1(x), J_1(x))$ are related to the maximal cuts of the two-loop sunrise integral.

The hierarchy of logarithms in eq.~\eqref{eqn:banzerolocsolall} allows us to read off the monodromy matrix:\footnote{Different than for the sunrise above, this time we are not going to normalize the matrix right away, for reasons to become clear below.}
\begin{align}
	\label{eqn:locmonzero}
	\rho_0(\gamma_{0})= 
\left(
\begin{array}{ccc}
 1 & -4 i \pi  & -4 \pi ^2 \\
 0 & 1 & -2 i \pi  \\
 0 & 0 & 1 \\
\end{array}
\right)\,.
\end{align}
The structure of the local solutions around the poles at $x_1=1/4$ and $x_2=1$ is different. For the singularity point $x_1=1/4$, the Frobenius method delivers the following functions: 
\beq\bsp
	\label{eqn:banquarterlocsol}
	\chi_3(x_1;x)&=1+4 \left(x-\qu\right)+\frac{64}{45} \left(x-\qu\right)^3-\frac{512}{945} \left(x-\qu\right)^4+\mathcal{O}\left(\left(x-\qu\right)^5\right)\,,\\
	\chi_2(x_1;x)&=
\sqrt{x-\qu}\left(1+2 \left(x-\qu\right)-2 \left(x-\qu\right)^2+\frac{548}{105} \left(x-\qu\right)^3\right.\\
		     &\qquad\qquad\qquad\qquad\qquad\qquad\left.-\frac{1306}{105} \left(x-\qu\right)^4+\mathcal{O}\left(\left(x-\qu\right)^5\right)\right)\,,\\
	\chi_1(x_1;x)&=
\left(x-\qu\right)\left(1+\frac{4}{3}
   \left(x-\qu\right)-\frac{16}{9} \left(x-\qu\right)^2+\frac{1088}{189} \left(x-\qu\right)^3\right.\\
		     &\qquad\qquad\qquad\qquad\qquad\qquad\left.
   -\frac{8704}{567} \left(x-\qu\right)^4+\mathcal{O}\left(\left(x-\qu\right)^5\right)\right)\,.
\esp\eeq
The structure of the solution space close to $x_2=1$ is similar. We see that the local exponents for $x_1$ and $x_2$ are $0,1/2,1$. Hence, the singular points $x_1$ and $x_2$ are not MUM-points. 
However, the corresponding monodromy matrices can be read off immediately also in this case: while the polynomials in $\chi_3$ and $\chi_1$ have trivial monodromy, the square root in $\chi_2(x_1;x)$ acquires a minus sign when transported around the singularity. Thus one finds: 
\begin{align}
	\label{eqn:locmonquarterone}
	\rho_{\qu}(\gamma_{\qu})=\rho_1(\gamma_{1})=\begin{pmatrix}1 & 0 & 0 \\ 0 & -1 & 0 \\ 0 & 0 & 1\end{pmatrix}. 
\end{align}
The singular point $x_3=\infty$ is a MUM-point. Substituting $x\to\frac{1}{y}$ in  $\cL^{\ban,(3)}_x$ leads to the differential operator 
\begin{equation}
	\cL^{\ban,(3)}_{1/y} = -y^3 \partial^3_y-\frac{3 (y (4 y-15)+8) y^2 }{2 (y-4) (y-1)}\partial^2_y-\frac{(y (7 y-17)+4) y }{(y-4) (y-1)}\partial_y-\frac{y}{y-4}
\end{equation}
Since $y=0$ is a MUM point, the structure of the solutions has a logarithmic hierarchy again, and the monodromy matrix equals the one in eq.~\eqref{eqn:banzerolocsolall}: $\rho_\infty(\gamma_{\infty})=\rho_0(\gamma_0)$.

Again, the matching matrices can be calculated using eq.~\eqref{eqn:calcmatchingmatrices}. However, while it was comparably easy to infer the analytic values of the matrix entries in the sunrise case, the banana matrices require rather high orders in the expansion of the solutions. Using expansions up to order 120 and the PSLQ~\cite{PSLQ} algorithm, as well as eq.~\eqref{eqn:monodromytranslation}, one obtains the monodromy matrices ($\mathrm{L}_2=\log(2)$):
\begin{small}
\begin{align}
	\label{eqn:banresult}
	\rho_0(\gamma_0)&=
\left(
\begin{array}{ccc}
 1 & -4 i \pi  & -4 \pi ^2 \\
 0 & 1 & -2 i \pi  \\
 0 & 0 & 1 \\
\end{array}
\right)\!,\\
\nonumber	\rho_0(\gamma_{\qu}) & =\frac{1}{\pi^2}
\left(
\begin{array}{ccc}
 12 \mathrm{L}_2^2 & 4 \mathrm{L}_2 \left(\pi ^2-12 \mathrm{L}_2^2\right) & \frac{1}{3} \left(\pi ^2-12 \mathrm{L}_2^2\right)^2 \\
 6 \mathrm{L}_2 & \pi ^2-24 \mathrm{L}_2^2 & -2 \mathrm{L}_2 \left(\pi ^2-12 \mathrm{L}_2^2\right) \\
 3 & -12 \mathrm{L}_2 & 12 \mathrm{L}_2^2 \\
\end{array}
\right)\!,\,\\
\nonumber	\rho_0(\gamma_1)&=\frac{1}{\pi^2}
\left(
\begin{array}{ccc}
	3 (4 \mathrm{L}_2{+}i \pi )^2 & -4 (4 \mathrm{L}_2{+}i \pi ) \left(12 \mathrm{L}_2^2{+}6 i \pi  \mathrm{L}_2{-}\pi ^2\right) & \frac{4}{3} \left(-12 \mathrm{L}_2^2-6 i \pi  \mathrm{L}_2+\pi
   ^2\right)^2 \\
		6 (4 \mathrm{L}_2{+}i \pi ) & -96 \mathrm{L}_2^2-48 i \pi  \mathrm{L}_2+7 \pi ^2 & 2 (4 \mathrm{L}_2{+}i \pi ) \left(12 \mathrm{L}_2^2{+}6 i \pi  \mathrm{L}_2{-}\pi ^2\right) \\
 12 & -12 (4 \mathrm{L}_2+i \pi ) & 3 (4 \mathrm{L}_2+i \pi )^2 \\
\end{array}
\right)\!,\,\\
\nonumber	\rho_0(\gamma_\infty)&=\frac{1}{\pi^2}
\left(
\begin{array}{ccc}
 4 (2 \pi -3 i \mathrm{L}_2)^2 & -12 i (\pi -2 i \mathrm{L}_2)^2 (2 \pi -3 i\mathrm{L}_2) & -9 (\pi -2 i \mathrm{L}_2)^4 \\
 -18 \mathrm{L}_2-12 i \pi  & 72\mathrm{L}_2^2+72 i \pi \mathrm{L}_2-17 \pi ^2 & 6 (3 \mathrm{L}_2+i \pi ) (\pi -2 i\mathrm{L}_2)^2 \\
 -9 & 12 (3 \mathrm{L}_2+i \pi ) & 4 (\pi -3 i\mathrm{L}_2)^2 \\
\end{array}
\right).
\end{align}
\end{small}
We can check that $\rho_0(\infty)=(\rho_0(1)\rho_0(\qu)\rho_0(0))^{-1}$. 
These matrices generate the monodromy group associated to the banana differential operator $\cL^{\ban,(3)}_x$ as a subgroup of $\textrm{GL}_3(\mathbb{C})$. It is possible to choose a basis for the solution space so that the monodromy matrices have integer entries. However, this will not be needed in our case. Instead, we want to use the fact that $\cL^{\ban,(3)}_x$ is a symmetric square to identify the image of the monodromy group in $\textrm{GL}_3(\mathbb{C})$ as arising from a group of $2\times2$ matrices. More precisely, we are looking for a set of $2\times2$ matrices such that acting with the monodromy matrix $\rho_0(0)$ on a $3$-dimensional solutions vector should equal the action of the $2\times 2$-representation on the solution vector $\begin{pmatrix}\Psi_2\\\Psi_1\end{pmatrix}$. Since the relation between the two- and three-dimensional solution spaces is most transparent from eq.~\eqref{eq:H1_to_Psi1}, we prefer to work here with the basis of the solution space of $\cL_x^{\mathsf{ban}}$ from eq.~\eqref{eq:H1_to_Psi1}. Due to the change of basis we need to conjugate the monodromy matrices in eq.~\eqref{eqn:banresult} with the matrix $B$ from eq.~\eqref{eqn:Brot} resulting in 
\begin{align}
	\label{eqn:BMon}
&\tilde{\rho}_0(\gamma_0)=\left(
\begin{array}{ccc}
 1 & 6 i & -5 \\
 0 & 1 & i \\
 0 & 0 & 1 \\
\end{array}
\right)\!,\quad\tilde{\rho}_0(\gamma_{\qu})=\left(
\begin{array}{ccc}
 1 & 0 & 0 \\
 -2 i & 3 & 2 i \\
 4 & 4 i & -3 \\
\end{array}
\right)\!,\\
&\tilde{\rho}_0(\gamma_1)\left(
\begin{array}{ccc}
 -3 & -10 i & 7 \\
 -12 i & 31 & 21 i \\
 16 & 40 i & -27 \\
\end{array}
\right)\,.\nonumber
\end{align}
The comparison is made in components, here for example the equation for the third component: 
\begin{equation}
	\left(\underbrace{B^{-1}\rho_0(0) B}_{\tilde{\rho}_0(\gamma_0)} \begin{pmatrix}I_1(x)\\J_1(x)\\H_1(x)\end{pmatrix}\right)_{\!\!3}\stackrel{!}{=}-\frac{1}{2}\,t\,\left(\Psi_1(t)^2\right)_{\circlearrowleft}\,,
\end{equation}
where we use the following ansatz for the monodromy matrices:
\begin{equation}
	\label{eqn:tbtansatz}
	\begin{pmatrix}\Psi_2(t)\\\Psi_1(t)\end{pmatrix}_{\circlearrowleft} = \begin{pmatrix}c_{11} & c_{12} \\ c_{21} & c_{22}\end{pmatrix}\begin{pmatrix}\Psi_2(t)\\\Psi_1(t)\end{pmatrix}\,.
\end{equation}
Plugging in a couple of numerical values for $t$, which are selected such as to place $x$ in the corresponding region, allows to determine the values $c_{ij}$ in the ansatz in eq.~\eqref{eqn:tbtansatz} for each generator. Finally, one finds the following representations for the generators of the monodromy group: 
\beq\bsp
	\label{eqn:banresult2}
	\cR_0:=\rho_0^\tbt(\gamma_0)    &=\begin{pmatrix}1 & -1 \\ 0 & 1 \end{pmatrix},\quad 
	\cR_{\qu}:=\rho_0^\tbt(\gamma_{\qu})  =-i{\sqrt{3}}\begin{pmatrix}1 & 2/3 \\ -2 & -1 \end{pmatrix},\\
	\cR_1:=\rho_0^\tbt(\gamma_1)    &=-i{\sqrt{3}}\begin{pmatrix}1 & 1/3 \\ -4 & -1 \end{pmatrix}\,.
\esp\eeq
These three matrices generate a subgroup $\Gamma^{\mathsf{ban},(2)}$ of $\textrm{GL}_2(\mathbb{C})$, which is closely related to the monodromy group $\Gamma^{\mathsf{ban},(3)}\subset\textrm{GL}_3(\mathbb{C})$ generated by the matrices $\tilde{\rho}_0(\gamma_0)$, $\tilde{\rho}_0(\gamma_{\qu})$ and $\tilde{\rho}_0(\gamma_1)$ in eq.~\eqref{eqn:BMon}. 
More precisely, consider the map $\sigma: \textrm{GL}_2(\mathbb{C}) \to \textrm{GL}_3(\mathbb{C})$ defined by
\beq
\sigma\left(\begin{smallmatrix}a& b\\c& d\end{smallmatrix}\right) = \frac{1}{3}\,\left(
\begin{smallmatrix}
 (a+c) (3 a+c) & 2 i \left(6 a^2-9 a b+8 a c-6 a d-6 b c+2 c^2-3 c d\right) & -3 (3 a-3 b+c-d) (a-b+c-d) \\
 i c (a+c) & -4 a c+3 a d+3 b c-4 c^2+6 c d & -3 i (c-d) (a-b+c-d) \\
 -c^2 & -2 i c (2 c-3 d) & 3 (c-d)^2 \\
\end{smallmatrix}
\right)\,.
\eeq
One can show that $\sigma$ is a group homomorphism with kernel $\textrm{Ker }\sigma = \mathbb{Z}_2$ such that
\beq
\tilde\rho_0 = \sigma\circ \rho_0^{2\times2}\,.
\eeq
Together with $-\mathds{1}\notin\Gamma^{\mathsf{ban},(2)}$, it follows that $\sigma(\Gamma^{\mathsf{ban},(2)}) = \Gamma^{\mathsf{ban},(3)}$, and so $\Gamma^{\mathsf{ban},(2)}$ and $\Gamma^{\mathsf{ban},(3)}$ are isomorphic. 

Let us discuss the structure of the group $\Gamma^{\mathsf{ban},(2)}$ in a bit more detail. First, one can check (e.g., by comparing to {\tt Sage}) that $\cR_0$, $\cR_0^{-1}\cR_{1}\cR_{\qu}\cR_0$, $\cR_0^{-1}\cR_{\qu}\cR_0\cR_{\qu}\cR_0$ and $-\mathds{1}$ are generators of $\Gamma_0(6)$. Note that $\cR_1$ and $\cR_{\qu}$ are self-inverse. We thus see that, while $\Gamma^{\mathsf{ban},(2)}$ does not contain $\Gamma_0(6)$ as a subgroup, it does contain\footnote{The notion $\overline{\Gamma}$ has been defined at the end of subsection~\ref{ssec:sunrisefamily}.} $\overline{\Gamma_0(6)}$. Moreover, one can easily check that 
\beq\label{eq:Rel_to_Verrill}
\overline{\Gamma^{\mathsf{ban},(2)}}\simeq \overline{\Gamma_0(6)^{+3}}\,,
\eeq
with
\beq\bsp
\Gamma_0(6)^{+3} &\,= \left\{\left(\begin{smallmatrix} a& b\\ 6c &d\end{smallmatrix}\right),\sqrt{3}\left(\begin{smallmatrix} a& b/3\\ 2c &d\end{smallmatrix}\right)\in\textrm{SL}_2(\mathbb{R})\big| a,b,c,d\in\mathbb{Z}\right\}\\
&\,=\Gamma_0(6) \cup (i\cR_{\qu})\Gamma_0(6)\,.
\esp\eeq

Next, let us discuss how modular forms and modular functions make an appearance here. We define (cf.~eq.~\eqref{eq:tau_def}):
\beq\bsp
\tau &\,= i\frac{J_1(x)}{H_1(x)}-1 = \frac{\Psi_2(t)}{\Psi_1(t)}\,.
\esp\eeq
We can invert this relation to express $x$ in terms of $\tau$~\cite{verrill1996}:
\beq
x(\tau) = -4\,\left(\frac{\eta(2\tau)\eta(6\tau)}{\eta(\tau)\eta(3\tau)}\right)^6\,.
\eeq
We also define:\footnote{Our definition of $\varpi(\tau)$ differs from the one used by Verrill in ref.~\cite{verrill1996} by a factor $(2\pi i)^2$, i.e., $\varpi^{\textrm{our}}(\tau) = (2\pi i)^2\varpi^{\textrm{Ver.}}(\tau)$.}
\beq
\varpi(\tau) = H_1(x(\tau)) = \frac{\eta(2\tau)^4\eta(6\tau)^4}{\eta(\tau)^2\eta(3\tau)^2}\,.
\eeq
Let us discuss the modular properties of $x(\tau)$ and $\varpi(\tau)$. One finds that $x(\tau)$ is a modular function and $\varpi(\tau)$ is a modular form of weight two for $\Gamma_0(6)$
Moreover, we find
\beq
x(\cR_{\qu}\cdot \tau) = x(\cR_{1}\cdot \tau) = x(\tau)\,,
\eeq
which shows that $x(\tau)$ is a modular function for the monodromy group $\Gamma^{\mathsf{ban},(3)}\simeq\Gamma^{\mathsf{ban},(2)}$, as expected. In addition, since $\Gamma^{\mathsf{ban},(2)}$ acts via M\"obius transformations via $\overline{\Gamma^{\mathsf{ban},(2)}}$, eq.~\eqref{eq:Rel_to_Verrill} implies that $x(\tau)$ is also a modular function for $\Gamma_0(6)^{+3}$. Similarly, we find:
\beq\bsp\label{eq:modular_varpi}
\varpi(\cR_{\qu}\cdot \tau) &\,= -3(2\tau+1)^2\,\varpi(\tau)\,,\\
\varpi(\cR_{1}\cdot \tau) &\,= -3(4\tau+1)^2\,\varpi(\tau)\,.
\esp\eeq
Accordingly, $\varpi(\tau)$ is a modular form of weight two for the monodromy group $\Gamma^{\mathsf{ban},(3)}\simeq\Gamma^{\mathsf{ban},(2)}$, again as expected. However, $\varpi(\tau)$ is not a modular form for $\Gamma_0(6)^{+3}$, which would require the factor of automorphy to be $+3(c\tau+d)$ in eq.~\eqref{eq:modular_varpi}.

\vspace{1mm}\noindent



\section{Banana integrals and iterated integrals of meromorphic modular forms}
\label{sec:bananameromorphic}

The analysis of the monodromy group of the sunrise and banana integrals implies via Theorem~\ref{thm:section2} that both integrals can be expressed through all orders in $\eps$ in terms of iterated integrals of meromorphic modular forms for the congruence subgroup $\Gamma_1(6)$, which is a neat subgroup in the sense of Definition~\ref{defi:neat}. In the remainder of this section we make this statement concrete, and we derive a form of the differential equation satisfied by the master integrals for the sunrise and banana families that involves the basis of meromorphic modular forms defined in Theorem~\ref{thm:main} only. 

The strategy of section~\ref{sec:FIs_and_deqs} of solving the first-order systems then implies that all iterated integrals that appear in the solution, to all orders in the dimensional regulator $\eps$, only involves the basis of meromorphic modular forms implied by Theorem~\ref{thm:main}. Note that once this form of the differential equation is known, it is straightforward to solve it explicitly. In particular, the initial condition is known to all orders in $\eps$ in terms of $\Gamma$ functions for banana integrals of arbitrary loop order~\cite{Bonisch:2021yfw}.

\subsection{The iterated integrals for the sunrise integral}
We start by discussing the case of the two-loop equal-mass sunrise integral. This case is in principle well known, and we will show that we can recover the results of ref.~\cite{Adams:2018yfj}. However, we discuss this case in some detail, as it allows us to set our conventions and to point out differences with respect to the three-loop equal-mass banana integral, which will be discussed in section~\ref{sec:ban_iterated}.

There are two master integrals $(\cS_1(\eps;t), \cS_2(\eps;t))$ for the two-loop equal-mass sunrise integral, which are accompanied by a tadpole integral (which is constant in our normalisation: $I^{\sun}_{2,2,0,0,0}(p^2,m^2;2-2\eps)=1$)\cite{Remiddi:2016gno}:
\beq\bsp\label{eq:sunrise_MIs}
\cS_1(\eps;t) &\, = -I^{\sun}_{1,1,1,0,0}(p^2,m^2;2-2\eps)\,,\\
\cS_2(\eps;t) &\, = -\left[\frac{1}{3}(t^2-6t+21)-12\eps(t-1)\right]\,I^{\sun}_{1,1,1,0,0}(p^2,m^2;2-2\eps)\\
&\,\phantom{=}-2(t-1)(t-9)\,I^{\sun}_{2,1,1,0,0}(p^2,m^2;2-2\eps)\,,
\esp\eeq
where the variable $t$ is defined in eq.~\eqref{eq:t_def}. Both master integrals are finite for $\eps=0$ and satisfy the differential equation:
\beq\label{eq:sun_GM}
\partial_t\begin{pmatrix}\cS_1(\eps;t) \\ \cS_2(\eps;t)\end{pmatrix} = \Big[B^\sun(t) -2 \eps D^\sun(t)\Big]\begin{pmatrix}\cS_1(\eps;t) \\ \cS_2(\eps;t)\end{pmatrix} + \begin{pmatrix}0\\1\end{pmatrix}\,.
\eeq
Explicit expression for the matrices $B^\sun(t)$ and $D^\sun(t)$ are collected in appendix~\ref{app:sun}.

In a next step, we want to introduce a modular parametrisation and apply the results of section~\ref{sec:mero_sec}. In order to do this, the Hauptmodul needs to be normalised as in eq.~\eqref{eqn:Hauptmodul}. This is, however, not the case for the Hauptmodul for $\Gamma_1(6)$ defined in eq.~\eqref{eq:t_def}. We define:
\beq\label{eq:t_to_xi}
\tz(\tau) = \frac{9}{t(\tau)} = \frac{\eta(2\tau)^8\eta(3\tau)^4}{\eta(\tau)^4\eta(6\tau)^8} = \frac{1}{q} + \ord(q^0)\,.
\eeq
Similarly, we define
\beq
 \mH_1(\tau)
 = \frac{\sqrt{3}}{2\pi\,\tz(\tau)}\,\Psi_1(t(\tau)) = \frac{\eta(\tau)\eta(6\tau)^6}{\eta(2\tau)^2\eta(3\tau)^3} = q + \ord(q^2)\,,
\eeq
in agreement with the normalisation in eq.~\eqref{eq:Hk_normalisation} for $k=1$.\footnote{Since $\Gamma_1(6)$ is neat, we must $h=1$ in eq.~\eqref{eq:Hk_normalisation}. Moreover, $\Gamma_1(6)$ has four cusps, so eq.~\eqref{eq:eps_infty_to_dGamma} implies $d_{\Gamma_1(6)}=1$.} The Jacobian of the change of variables from $\tz$ to $\tau$ is
\beq
d\tz = -2\pi i \tz(\tau)(\tz(\tau)-1)(\tz(\tau)-9)\,\mH_1(\tau)^2\,d\tau\,.
\eeq
Then, letting
\beq
\begin{pmatrix}\cS_1(\eps;t) \\ \cS_2(\eps;t)\end{pmatrix}   = \frac{1}{\eps^2(2\pi i)^2\,\tz(\tz-1)(\tz-9)}\,W^{\sun}(\tau)\begin{pmatrix}\widetilde{\cS}_1(\eps;t) \\ \widetilde{\cS}_2(\eps;t)\end{pmatrix}\,,
\eeq
with
\beq
W^{\sun}(\tau) = \begin{pmatrix} (2\pi i)^2 \tz(\tz-1)(\tz-9)\,\mH_1(\tau) & 0 \\
\frac{\pi^2}{3}\left[11 \tz^2-54\tz+27+6\eps(\tz+3)^2\right]\,\mH_1(\tau)-\frac{G_2(\tau)}{\mH_1(\tau)} & -\frac{2\pi i \eps}{\mH_1(\tau)}
\end{pmatrix}\,,
\eeq
we find
\beq
\partial_{\tau}\begin{pmatrix}\widetilde{\cS}_1(\eps;t) \\ \widetilde{\cS}_2(\eps;t)\end{pmatrix} = \eps\,\widetilde{D}^\sun(\tau)\,\begin{pmatrix}\widetilde{\cS}_1(\eps;t) \\ \widetilde{\cS}_2(\eps;t)\end{pmatrix}+108\pi^2 \eps\,(\tz-1)(\tz-9)\,\mH_1(\tau)^3\,\begin{pmatrix}0\\1\end{pmatrix}\,,
\eeq
with
\beq
\widetilde{D}^\sun(\tau) = \begin{pmatrix}i\pi (\tz^2+10\tz-27) \,\mH_1(\tau)^2 & 1 \\ 
-\pi^2\,(\tz+3)^4\,\mH_1(\tau)^4 & i\pi (\tz^2+10\tz-27) \,\mH_1(\tau)^2
\end{pmatrix}\,.
\eeq
In the previous equations we used the shorthand $\tz=\tz(\tau)$ to keep the notation as light as possible. Let us comment on the form of the differential equation. We observe that the differential equation only involves (holomorphic) modular forms of weights up to 4 for $\Gamma_1(6)$. We also observe that $\eps$ factorises from the matrix multiplying the homogeneous part, so that the differential equation is in canonical form. As a consequence, the master integrals in the basis $(\widetilde{\cS}_1(\eps;t), \widetilde{\cS}_2(\eps;t))$ can be expressed in terms of iterated integrals of modular forms for $\Gamma_1(6)$, which are pure functions of uniform weight~\cite{ArkaniHamed:2010gh} according to the definition of ref.~\cite{Broedel:2018qkq}. The initial condition can be fixed to all orders in $\eps$ in terms of zeta values. These results are actually not new, but they agree with the findings of ref.~\cite{Adams:2018yfj}. The change of basis from $(\widetilde{\cS}_1(\eps;t), \widetilde{\cS}_2(\eps;t))$ to $({\cS}_1(\eps;t), {\cS}_2(\eps;t))$ involves a matrix whose entries are rational in $\eps$ and meromorphic quasi-modular forms for $\Gamma_1(6)$ with poles at most at the cusps. More precisely, for $i\ge j$, we have
\beq\bsp
\widetilde{D}^\sun(\tau)_{ij} &\,\in M_{2(1+i-j)}(\Gamma_1(6))\,,\\
W^{\sun}(\tau)_{ij} &\,\in \cQ\cM_{3-2j}^{\le (i-1)}(\Gamma_1(6),S_{\Gamma_1(6)})(\eps)\,.
\esp\eeq


\subsection{The iterated integrals for the three-loop banana integral}
\label{sec:ban_iterated}

We now extend the discussion of the previous section to the three-loop equal-mass banana integrals. We choose three master integrals as~\cite{Primo:2017ipr}
\eqs{
    \label{eq:fints}
\cI_1(\eps;x) &= (1 + 2 \epsilon ) (1 + 3 \epsilon)I_{1,1,1,1,0,0,0,0,0}(p^2,1;2-2\eps)\, ,\\
\cI_2(\eps;x) &= (1 + 2 \epsilon ) I_{2,1,1,1,0,0,0,0,0}(p^2,1;2-2\eps)\, ,\\
\cI_3(\eps;x) &= I_{2,2,1,1,0,0,0,0,0}(p^2,1;2-2\eps)\,.
}
All three master integrals are finite at $\eps=0$. 
The fourth master integral is the three-loop tadpole integral with squared propagators (which in our
normalisation again evaluates to unity, $I_{2,2,2,0,0,0,0,0,0}(p^2,1;2-2\eps) = 1$).
The three master integrals in eq.~\eqref{eq:fints} satisfy the inhomogeneous equation~\cite{Primo:2017ipr}
\begin{equation}\label{eq:banana_DEQ}
\partial_x \ivec = \Big[B^\ban(x) + \eps D^\ban(x)\Big]\ivec + \begin{pmatrix}0\\0\\-\frac{1}{2(4x-1)}\end{pmatrix}\,.
\end{equation}
The explicit expressions of the matrices can be found in appendix~\ref{app:ban}.

We change variables from $x$ to $t$ according to eq.~\eqref{eq:change_of_vars}, followed by the change of variables in eq.~\eqref{eq:t_to_xi}. We introduce a new basis according to
\beq
\ivec = \frac{(1+2\eps)(1+3\eps)}{\eps^2}\,W^\ban(\tau)\begin{pmatrix}\widetilde{\cI}_1(\eps;\tau)\\\widetilde{\cI}_2(\eps;\tau)\\\widetilde{\cI}_3(\eps;\tau)\end{pmatrix}\,.
\eeq
The non-vanishing entries of $W^\ban(\tau)$ are:
\begin{align}
\nonumber W^{\ban}&(\tau)_{11} = (2\pi i)^2\, \tz\,\mH_1(\tau)^2\,,\\
W^{\ban}&(\tau)_{21} = 
\frac{\tz }{2 \left(\tz^2-9\right) (1+3 \epsilon )}\,G_2(\tau )\\
\nonumber&\,+\frac{ \left(\tz^2-12 \tz+27\right) (\tz+3)^2+6\eps \left(\tz^4+20 \tz^3-90 \tz^2+180 \tz+81\right)}{6 \left(\tz^2-9\right)^2 (1+3 \epsilon )}\,\pi ^2 \tz\,\mH_1(\tau )^2\,,\\
\nonumber W^{\ban}&(\tau)_{22} = \frac{\pi\,\eps\,\tz}{2(1+3\eps)\,(\tz^2-9)}\,,\\
 \nonumber W^{\ban}&(\tau)_{31} = -\frac{\tz }{24 \pi ^2 (\tz-9) (\tz-1) (\tz+3)^2 (1+2 \epsilon ) (1+3 \epsilon)}\,\frac{G_2(\tau )^2}{ \mH_1(\tau )^2}\\
\nonumber&-\frac{\tz^4+6 \tz^3-540 \tz^2 +162 \tz+243+6\eps \left(\tz^4+16 \tz^3-306 \tz^2+144 \tz+81\right) }{36 (\tz-9) (\tz-3) (\tz-1) (\tz+3)^3 (1+2 \epsilon) (1+3 \epsilon)}\,\tz\,G_2(\tau )\\
\nonumber&-\frac{\pi ^2 \tz \mH_1(\tau )^2 }{216 (\tz-9) (\tz-3)^2 (\tz-1) (\tz+3)^4 (1+2 \epsilon) (1+3 \epsilon)}\\
\nonumber&\quad\times\Big[(\tz^5+3 \tz^4+1062 \tz^3-3726 \tz^2+729 \tz+2187) (\tz+3)^3\\
\nonumber&\quad\phantom{\times}+12\eps (\tz^8+10 \tz^7+1386 \tz^6-18126 \tz^5+82188 \tz^4-194562 \tz^3\\
\nonumber&\quad \phantom{\times}+78246 \tz^2+39366 \tz+19683)  +36 \epsilon ^2 (\tz^8+40 \tz^7+860 \tz^6-17064 \tz^5+68454 \tz^4\\
\nonumber&\quad\phantom{\times} -153576 \tz^3+69660 \tz^2+29160 \tz+6561)  \Big]\,,\\
\nonumber W^{\ban}&(\tau)_{32} = -\frac{\tz \epsilon  }{12 \pi  (\tz-9) (\tz-1) (\tz+3)^2 (1+2 \epsilon ) (1+3 \epsilon )}\,\frac{G_2(\tau )}{ \mH_1(\tau )^2}\\
\nonumber&-\frac{\pi  \tz \epsilon  \left[\tz^4+6 \tz^3-540 \tz^2 +162 \tz+243+6\epsilon\, \left(\tz^4+16 \tz^3-306 \tz^2+144 \tz+81\right) \right]}{36 (\tz-9) (\tz-3) (\tz-1) (\tz+3)^3 (1+2 \epsilon ) (1+3 \epsilon )}\,,\\
\nonumber W^{\ban}&(\tau)_{33} = -\frac{\tz \epsilon ^2}{2 (\tz-9) (\tz-1) (\tz+3)^2 (1+2 \epsilon) (1+3 \epsilon) \mH_1(\tau )^2}\,.
\end{align}
Note that the entries of $W^{\ban}(\tau)$ are again rational in $\eps$ and meromorphic quasi-modular forms for $\Gamma_1(6)$:
\beq
W^{\ban}(\tau)_{ij} \in \cQ\cM_{4-2j}^{\le(i-1)}(\Gamma_1(6),S_{\Gamma_1(6)}\cup \{[\tau_{\pm 3}]\})(\eps)\,, \qquad \tz(\tau_{\pm 3}) = \pm 3\,.
\eeq
Unlike the case of the two-loop sunrise integral, now we do not only have poles at the MUM-points $\tz\in \{0,1,9,\infty\}$, but we also have poles at $\tz = \pm 3$. These poles arise from the singularities of the differential operator in eq.~\eqref{eq:L_ban_3} which are not MUM-points, i.e., $x\in\{1/4,1\}$.

The vector $(\widetilde{\cI}_1(\eps;\tau),\widetilde{\cI}_2(\eps;\tau),\widetilde{\cI}_3(\eps;\tau))$ satisfies the differential equation:
\beq\bsp\label{eq:ban_final}
\partial_{\tau}\begin{pmatrix}\widetilde{\cI}_1(\eps;\tau)\\\widetilde{\cI}_2(\eps;\tau)\\\widetilde{\cI}_3(\eps;\tau)\end{pmatrix}
&\,=i\,\eps\,\widetilde{D}^{\ban}(\eps;\tau)\begin{pmatrix}\widetilde{\cI}_1(\eps;\tau)\\\widetilde{\cI}_2(\eps;\tau)\\\widetilde{\cI}_3(\eps;\tau)\end{pmatrix} + 8\pi i (\tz-1)(\tz-9)(\tz^2-9)\,\mH_1(\tau)^4\begin{pmatrix}0\\0\\1\end{pmatrix} \,,
\esp\eeq
with
\beq
\widetilde{D}^{\ban}(\eps;\tau) = \begin{pmatrix} d_2(\tau)  & -1  &0\\
d_4(\tau) &  d_2(\tau)& -6 \\
\frac{1-4\eps^2}{\eps^2}\,\,d_6(\tau) & \frac{1}{6}d_4(\tau) & d_2(\tau)
\end{pmatrix}\,,
\eeq
where we defined:
\beq\bsp
d_2(\tau) &\,=\frac{4\pi\,(\tz^4-10\tz^3+18\tz^2-90\tz+81)}{\tz^2-9}\,\mH_1(\tau)^2\,,\\
d_4(\tau) &\,=-\frac{2\pi^2\,(\tz^2-18\tz+9)^2\,(\tz^4-12\tz^3+102\tz^2-108\tz+81)}{(\tz^2-9)^2}\,\mH_1(\tz)^4\,,\\
d_6(\tau) &\,=\frac{8\pi^3\,\tz\,(\tz^2-18\tz+9)^3\,(\tz^4-12\tz^3+38\tz^2-108\tz+81)}{3(\tz^2-9)^3}\,\mH_1(\tau)^6\,.
\esp\eeq
The structure of the differential equation is particularly simple, and the functions $d_k(\tau)$ are meromorphic modular forms of weight $k$:
\beq
d_k(\tau)\in \widetilde{\cM}_{k}(\Gamma_1(6), \{[\tau_{\pm3}]\})\,.
\eeq
The appearance of the additional poles at $\tz = \pm3$ can again be traced back to the singularities at $x\in\{1/4,1\}$, which are not MUM-points, and so they do not map to cusps when passing to the variable $\tau$.
We note that in order to arrive at this simple form, the algorithm of section~\ref{sec:neat_proof}, which allows every meromorphic quasi-modular form to be decomposed according to Theorem~\ref{thm:main}, plays a crucial role. The differential equation can easily be solved to arbitrary orders in $\eps$ in terms of iterated integrals of meromorphic modular forms for $\Gamma_1(6)$. The initial condition is known to all order in $\eps$ from ref.~\cite{Broedel:2019kmn,Bonisch:2021yfw}. We have explicitly computed all master integrals through $\ord(\eps^2)$, and we have checked numerically that our results are correct by comparing the numerical evaluation of the iterated integrals in terms of $q$-expansions to a direct numerical evaluation of the banana integrals from Mellin--Barnes integrals in the Euclidean region. The results are lengthy and not very illuminating, and they are available from the authors upon request. 

Let us conclude by making an important observation. Despite all the structural similarities between $\widetilde{D}^{\sun}(\tau)$ and $\widetilde{D}^{\ban}(\tau)$, the differential equation~\eqref{eq:ban_final} is \emph{not} in canonical from, because the entry in the lower left corner of $\widetilde{D}^{\ban}(\tau)$ is not independent of $\eps$! This is not entirely surprising: Canonical differential equations are expected to be closely related to the concept of pure functions~\cite{Henn:2013pwa}. Pure functions in turn are expected to have only logarithmic singularities~\cite{ArkaniHamed:2010gh,Broedel:2018qkq}. We see, however, that $d_4(\tau)$ and $d_6(\tau)$ have double and triple poles at $\tz=\pm3$. More generally, we see that, as soon as we consider poles that do not lie at the cusps, the basis of meromorphic modular forms obtain from Theorem~\ref{thm:main} will generically lead to functions with higher-order poles, and there is in general no way to preserve modularity and only have logarithmic singularities. 
It is possible to achieve an alternative decomposition which leads to a basis of quasi-modular forms with single-poles. More precisely, for $k\ge 2$ we have a decomposition:
\beq\label{eq:dec_QM}
\cQ\cM_{k-2}(\Gamma,R_S) = \delta\cQ\cM_{k-2}(\Gamma,R_S) \oplus \cM_{2-k}(\Gamma,R_S)\,G_2^{k-1}\oplus \widetilde{\cQ\cM}_k(\Gamma,R_{\infty})\,,
\eeq
where we defined (cf.~eqs.~\eqref{eq:Mtilde_def} and~\eqref{eq:Shat_def}):
\beq\bsp\widetilde{\cQ\cM}_k(\Gamma,R_{\infty})&\, := M_k(\Gamma) \cup \widehat{\cS}_k(\Gamma)\cup\widehat{\cQ\cM}_k(\Gamma,R_{\infty})\,\\
\widehat{\cQ\cM}_k(\Gamma,R_{\infty})&\,:= \bigoplus_{\substack{P\in R \\ 0\le m<k-1}}\!\!\!\! \mathbb{C}\,\frac{\mH_{k-2m}\,G_2^m}{\tz-P}\,.
\esp\eeq
The difference between the sets $\widehat{\cM}_k(\Gamma,R_{\infty})$ and $\widehat{\cQ\cM}_k(\Gamma,R_{\infty})$ is that the former only contains meromorphic modular forms, but with poles of higher order, and the latter only contains quasi-modular forms of higher depth, but with at most simple poles. The proof of the decomposition in eq.~\eqref{eq:dec_QM} (for neat subgroups of genus zero) is similar to the proof in section~\ref{sec:neat_proof}.
The form of the differential equation in this basis, however, is extremely complicated (and even further away from being canonical). For the future, it would be interesting to investigate if it is possible to to define a canonical basis for the equal-mass three-loop banana integrals. This may involve introducing integration kernels that are primitives of modular forms with only simple poles, but at the expense of loosing modularity, similar to the case of elliptic polylogarithms~\cite{Broedel:2017kkb} (see also ref.~\cite{matthesfonseca}).

\vspace{1mm}\noindent



\section{Conclusion}
\label{sec:conclusions}
In this paper we have considered a class of differential equations which can be solved to all orders in $\eps$ in terms of iterated integrals of meromorphic modular forms. We have described these differential equations in detail, and we have argued that the type of modular forms required is related to the monodromy group of the associated homogeneous differential equation. On the mathematical side, one of the main results of this paper is a generalisation of the main theorems for the full modular group $\slz$ of ref.~\cite{matthes2021iterated} to arbitrary genus-zero subgroups of finite index. In particular, we have provided an explicit decomposition of the space of meromorphic modular forms into a direct sum of two spaces. The first space collects all those meromorphic modular forms which can be written as derivatives of other functions, and which are thus irrelevant when considering integrals. We provide an explicit basis for the second space (at least in the case of so-called neat subgroups), and, using a classical result due to Chen, we show that the resulting iterated integrals are independent. 

On the physics side, we have clarified by explicit calculations how the monodromy groups of the associated homogeneous differential equations determine the type of modular forms that can arise. In particular, this gives another argument why the congruence subgroup associated to the two-loop equal-mass sunrise integral should have level 6, rather than 12 (see refs.~\cite{Adams:2017ejb,Frellesvig:2021vdl}). Finally, we have provided, for the first time, a complete description of the higher orders in $\eps$ for all master integrals for the three-loop equal-mass banana family. The results, which involve iterated integrals of meromorphic modular forms, are rather lengthy, and they are available from the authors upon request.

In some sense, the differential equations and iterated integrals considered here can be interpreted as one of the simplest generalisations of MPLs: while MPLs arise from iterated integrations of rational functions, our integrals arise from iterated integrations of rational functions multiplied by solutions of a second-order linear differential operator that admits a modular parametrisation. Moreover, similar to the case of MPLs, we can identify classes of differential equations which can always be solved in terms of these functions in an algorithmic way. Note that there are natural generalisations of the class of Feynman integrals to which our construction applies, like those where the maximal cuts are rational functions, but the inhomogeneity involves iterated integrals of meromorphic modular forms. This is for example the case for the two-loop kite integral or some integrals contributing to the three-loop $\rho$ parameter~\cite{Adams:2016xah,Remiddi:2016gno,Adams:2018yfj,Abreu:2019fgk}.

There are still some open questions. First, it often happens for Feynman integrals that one needs to consider additional square roots, in addition to modular forms (cf., e.g., ref.~\cite{Aglietti:2004tq}). If all square roots can be rationalised, one can reduce the complexity again to the situation of rational functions. In the setup of modular forms, however, if the branch points of the square root are not aligned with the cusps of the modular curve, it is not clear that the functions obtained by rationalising the square roots will fall within the class of meromorphic modular forms considered here. Second, we have shown that for the three-loop banana integrals, the differential equation is very compact when expressed in terms of the basis of meromorphic modular forms defined in section~\ref{sec:mero_sec}, but it is not in canonical form. For the future, it would be interesting to understand if and how a canonical form for this differential equation can be obtained, and what the resulting concept of pure functions would be. We leave these questions for future work.

\vspace{1mm}\noindent

\subsection*{Acknowledgments}
This project has received funding from the European Research Council (ERC) under the European Union's Horizon 2020 research and innovation programme (grant agreement No.~724638). JB is grateful to Pietro Longhi for discussions. Furthermore, the authors would like to thank Helena Verrill for correspondence. 

\appendix

\section{The case of a general finite-index subgroup of genus zero}
\label{app:mathy}

The purpose of this appendix is to give a proof of Theorem \ref{thm:main} for a general finite-index subgroup of genus zero. The difference to the case considered in section \ref{sec:mero_sec} is that now there might be elliptic points or irregular cusps, whose local analytic structure is more complicated.

Throughout this appendix, we keep the notation of section \ref{sec:mero_sec}.

\subsection{The case \texorpdfstring{$k\leq 1$}{}}

\begin{proposition} \label{prop:kleq1}
For $k\leq 1$, we have
\[
\cQ\cM_k(\Gamma,R_S)=\delta(\cQ\cM_{k-2}(\Gamma,R_S))\oplus\cM_k(\Gamma,R_S).
\]
In other words, the following statements are true.
\begin{itemize}
\item[(i)]
We have
\[
\delta(\cQ\cM_{k-2}(\Gamma,R_S))\cap \cM_k(\Gamma,R_S)=\{0\},
\]
as subspaces of $\cQ\cM_k(\Gamma,R_S)$.
\item[(ii)]
We have
\[
\cQ\cM_k(\Gamma,R_S)=\delta(\cQ\cM_{k-2}(\Gamma,R_S))+\cM_k(\Gamma,R_S).
\]
\end{itemize}

\end{proposition}
\begin{proof}
To prove (i), we need to show that, if $f\in \delta(\cQ\cM_{k-2}(\Gamma,R_S))\cap \cM_k(\Gamma,R_S)$, then $f=0$. The proof is essentially the same as in ref.~\cite[Theorem 6.1]{matthes2021iterated}, so we will omit some details. Let $g\in \cQ\cM_{k-2}(\Gamma,R_S)$ be such that $\delta(g)=f$, and denote by $g_0,\ldots,g_p$ the coefficient functions of $g$. The coefficient functions of $f$ are then given by
\[
f_r=\delta(g_r)+\frac{k-2-r+1}{12}g_{r-1}, \qquad 0\leq r\leq p+1,
\]
with the convention that $g_{-1}=g_{p+1}\equiv 0$. On the other hand, since $f$ is modular, we have $f_r=0$ for $1\leq r\leq p+1$. In particular, for $r=p+1$, we have $\frac{k-2-p}{12}g_p=0$, hence $g_p=0$ (here, we use that $k\leq 1$). By recursion on $p$, the same argument yields that $g_r=0$ for all $0\leq r\leq p$, so that $g=0$, by uniqueness of the coefficient functions.

For the proof of (ii), we need to show that every $f\in \cQ\cM_k(\Gamma,R_S)$ can be written as $f=\delta(g)+h$, for some $g\in \cQ\cM_{k-2}(\Gamma,R_S)$ and some $h\in \cM_k(\Gamma,R_S)$. Let $f_0,\ldots,f_p$ denote the coefficient functions of $f$ and assume without loss of generality that $f_p\neq 0$. We can write (cf.~ref.~\cite[Theorem 4.1]{Royer})
\[
f=\sum_{r=0}^p \overline{f}_r\cdot E_2^r,
\]
for uniquely determined $\overline{f}_r\in \cM_{k-2r}(\Gamma,R_S)$, where $E_2$ denotes the normalized, holomorphic Eisenstein series of weight two, and the integer $p$ is, by definition, the \emph{depth} of $f$. Moreover, we have $f_p=\overline{f}_p$. We now prove the desired statement by induction on $p$, the case $p=0$ being trivial (take $g=0$ and $h=f=\overline{f}_0$). In the general case, a direct computation shows that the meromorphic quasi-modular form
\[
\overline{f}_p\cdot E_2^p-\frac{12}{k-p-1}\delta(\overline{f}_p\cdot E_2^{p-1})
\]
has depth $\leq p-1$, and we conclude by the induction hypothesis.
\end{proof}

\subsection{The case \texorpdfstring{$k\geq 2$}{}}

\begin{proposition}
For $k\geq 2$, we have
\[
(\delta(\cQ\cM_{k-2}(\Gamma,R_S))+\cM_k(\Gamma,R_S))\oplus \cM_{2-k}(\Gamma,R_S)E_2^{k-1}=\cQ\cM_k(\Gamma,R_S).
\]
\end{proposition}

\begin{proof}
We begin by proving that
\[
(\delta(\cQ\cM_{k-2}(\Gamma,R_S))+\cM_k(\Gamma,R_S))\cap \cM_{2-k}(\Gamma,R_S)E_2^{k-1}=\{0\}.
\]
Let $g\in \cQ\cM_{k-2}(\Gamma,R_S)$ with coefficient functions $g_0,\ldots,g_p$ and assume that $f+\delta(g)=h\cdot E_2^{k-1}$, for some $f\in \cM_k(\Gamma,R_S)$ and $h\in \cM_{2-k}(\Gamma,R_S)$. Then
\[
\delta(g_r)+\frac{k-2-r+1}{12}g_{r-1}=0, \qquad \mbox{for all }r\geq k,
\]
which shows that $g$ necessarily has depth $\leq k-2$. Moreover, since $g$ has weight $k-2$, one can show that $\delta(g)$ has depth $\leq k-2$. On the other hand, $h\cdot E_2^{k-1}$ has depth $k-1$, unless $h=0$, so that the equality $f+\delta(g)=h\cdot E_2^{k-1}$ yields that $h=0$. Therefore, also $g=0$, as was to be shown.

We next show that 
\[
\delta(\cQ\cM_{k-2}(\Gamma,R_S))+\cM_k(\Gamma,R_S)+\cM_{2-k}(\Gamma,R_S)E_2^{k-1}=\cQ\cM_k(\Gamma,R_S).
\]
Let $f\in \cQ\cM_k(\Gamma,R_S)$ with coefficient functions $f_0,\ldots,f_p$, such that $f_p\neq 0$, and write $f=\sum_{r=0}^p\overline{f}_r\cdot E_2^r$, for some $\overline{f}_r\in \cM_{k-2r}$. As remarked above, we have $f_p=\overline{f}_p$. If $p\neq 0,k-1$, then the same argument as in the proof of Proposition \ref{prop:kleq1} shows that $\overline{f}_p\cdot E_2^p -\frac{12}{k-p-1}\delta(\overline{f}_p\cdot E_2^{p-1})$ has depth $\leq p-1$. The desired statement now follows by descending induction on $r$.
\end{proof}

\begin{proposition}
For $k\geq 2$, we have
\[
\delta(\cQ\cM_{k-2}(\Gamma,R_S))\cap \cM_k(\Gamma,R_S)=\delta^{k-1}(\cM_{2-k}(\Gamma,R_S)).
\]
\end{proposition}
\begin{proof}
Again, the proof is essentially the same as in ref.~\cite[Theorem 6.1]{matthes2021iterated}. Let $g\in \cQ\cM_{k-2}(\Gamma,R_S)$ be such that $\delta(g)=f$, and denote by $g_0,\ldots,g_p$ the coefficient functions of $g$. We may assume without loss of generality that $g_p\neq 0$. As in the case $k<2$, we have
\[
f_r=\delta(g_r)+\frac{k-2-r+1}{12}g_{r-1}, \qquad 0\leq r\leq p+1,
\]
and $f_r=0$, for $1\leq r\leq p+1$. In particular, the equality $f_{p+1}=0$ implies that $p=k-2$. Likewise, by recursion on $r$, the equality $f_r=0$ shows that $\delta(g_r)=-\frac{k-r-1}{12}g_{r-1}$, for all $1\leq r\leq p$. The desired statement now follows from $g=g_0$ and $g_p\in \cM_{2-k}(\Gamma,R_S)$, which are both general facts about the coefficient functions of quasi-modular forms (cf.~ref.~\cite{Royer}).
\end{proof}
In order to conclude the proof of Theorem \ref{thm:main} for arbitrary genus zero subgroups in the case $k\geq 2$, it now suffices to prove the following analogue of Theorem \ref{thm:bijection}.
\begin{thm} \label{theorem}
There is a direct sum decomposition
\[
\mathcal{M}_k(\Gamma,R_S)=\delta^{k-1}(\mathcal{M}_{2-k}(\Gamma,R_S))\oplus \widetilde{\mathcal{M}}_k(\Gamma,R_{s_0}).
\]
\end{thm}

\subsection{Divisors of meromorphic modular forms in negative weight}
The key ingredient for the proof of Theorem \ref{theorem} is a formula for the degree of the divisor
\[
\left\lfloor\operatorname{div} f\right\rfloor:=\sum_{P\in X_\Gamma}\left\lfloor \nu_P(f) \right\rfloor \cdot [P] \in \operatorname{Div}(X_\Gamma),
\]
where $0\neq f \in \mathcal{M}_{2-k}(\Gamma)$. Since the vanishing order of $f$ at an elliptic point or irregular cusp might be half- or third-integral, the divisor $\left\lfloor\operatorname{div} f\right\rfloor$ is in general different from $\operatorname{div}(f)=\sum_{P\in X_\Gamma}\nu_P(f) \cdot (P)$.
\begin{proposition} \label{proposition}
We have
\[
\deg \left\lfloor\operatorname{div} f\right\rfloor=\begin{cases}
0 & k=2\\
-1-\dim S_k(\Gamma)&k\geq 3.
\end{cases}
\]
\end{proposition}
\begin{proof}
The valence formula yields that
\begin{equation} \label{equation}
\deg \operatorname{div}(f)=\frac{(2-k)d_\Gamma}{12}=-(2-k)-\left(\frac{k}{4}-\frac 12\right)\varepsilon_2-\left(\frac{k}{3}-\frac 23\right)\varepsilon_3-\left( \frac{k}{2}-1 \right)\varepsilon_\infty,
\end{equation}
where in the second equality we have also used that $X_\Gamma$ has genus zero. This proves the statement for $k=2$, since in that case we necessarily have $\operatorname{div}(f)=\left\lfloor\operatorname{div} f\right\rfloor$. If $k\geq 3$, we first need a lemma which provides some arithmetic information about the vanishing order of $f$ at an elliptic point or an irregular cusp.
\begin{lemma} \label{lemma}
Let $k$ be an integer and $0\neq g\in \mathcal{M}_k(\Gamma)$. Then
\[
\nu_P(g) \equiv \begin{cases}
\frac{k}{4} \mod \mathbb Z & \mbox{if $P$ is elliptic of order two,} \\
\frac{k}{3} \mod \mathbb Z & \mbox{if $P$ is elliptic of order three,}\\
\frac{k}{2} \mod \mathbb Z & \mbox{if $P$ is an irregular cusp}.
\end{cases}
\]
\end{lemma}
\begin{proof}[Proof of Lemma \ref{lemma}.]
	The statement is trivial for $k=0$, since $g$ is then a holomorphic function on $X_\Gamma$ and therefore $\nu_P(g)\in \mathbb Z$. In general, if $\Gamma=\operatorname{SL}_2(\mathbb Z)$, then the desired assertion follows immediately by comparing both sides of the valence formula (where $i$ and $\rho$ denote the elliptic points defined in eq.~\eqref{eq:eps_infty_to_dGamma} above) 
\[
\nu_{[i]}(g)+\nu_{[\rho]}(g)+\sum_{P\in X_\Gamma\setminus \{[i],[\rho]\}}\nu_P(g)=\frac{k}{12}.
\]
For general $\Gamma$, if $k$ is even, then we can write $g=g'g''$ where $g'\in \mathcal{M}_0(\Gamma)$ and $g''\in \mathcal{M}_k(\operatorname{SL}_2(\mathbb Z))$, and the desired statement follows from the above, as every elliptic point is $\operatorname{SL}_2(\mathbb Z)$-equivalent to either $i$ or $\rho$, and every cusp is $\operatorname{SL}_2(\mathbb Z)$-equivalent to $\infty$. If $\Gamma$ is arbitrary and $k$ is odd, then the existence of a non-zero meromorphic modular form of weight $k$ implies that $-I\notin\Gamma$, hence that there are no elliptic points of order two. On the other hand, if $P$ is elliptic of order three, then $\nu_P(g)=\frac{1}{2}\nu_P(g^2) \equiv \frac{k}{3} \mod \mathbb Z$, by what was just established in the case of even weights. This proves the statement for elliptic points. Finally, if $P$ is an irregular cusp of width $h$, then $g(\tau+h)=(-1)^kg(\tau)$. On the other hand, the Fourier coefficients $\sum_{m=n}^\infty a_me^{\pi im\tau/h}$ of $g$ at $P$ satisfy $a_m=(-1)^ma_m$, and the result follows.
\end{proof}
We now return to the proof of Proposition \ref{proposition}. If $k$ is even, then combining \eqref{equation} with Lemma \ref{lemma} yields that
\[
\deg \left\lfloor\operatorname{div} f\right\rfloor=-(2-k)-\left\lfloor \frac{k}{4} \right\rfloor\varepsilon_2-\left\lfloor \frac{k}{3} \right\rfloor\varepsilon_3-\left(\frac{k}{2}-1\right)\varepsilon_\infty=-1-\dim S_k(\Gamma),
\]
proving the desired statement in that case.
If $k$ is odd, then a similar argument yields that
\[
\deg \left\lfloor\operatorname{div} f\right\rfloor=-(2-k)-\left\lfloor \frac{k}{3} \right\rfloor\varepsilon_3-\left(\frac{k}{2}-1\right)\varepsilon^{\rm reg}_\infty-\left(\frac{k}{2}-\frac 12\right)\varepsilon^{\rm irr}_\infty=-1-\dim S_k(\Gamma),
\]
where $\varepsilon^{\rm reg}_\infty$ (respectively, $\varepsilon^{\rm irr}_\infty$) denotes the number of regular (respectively, irregular) cusps. This ends the proof of Proposition \ref{proposition}.
\end{proof}
\begin{rmk}
If $\Gamma$ is an arbitrary finite-index subgroup of $\operatorname{SL}_2(\mathbb Z)$, not necessarily of genus zero, and $0\neq f\in \mathcal{M}_{2-k}(\Gamma)$, then essentially the same proof yields that
\[
\deg \left\lfloor\operatorname{div} f\right\rfloor=
\begin{cases}
0, & k=2,\\	
g-1-\dim S_k(\Gamma), & k\geq 3,
\end{cases}
\]
where $g$ denotes the genus of $X_\Gamma$. We expect that this formula is well-known to the experts but did not find it in the literature.
\end{rmk}

\subsection{Proof of \texorpdfstring{Theorem \ref{theorem}}{}}
Theorem \ref{theorem} follows by combining the assertions in the next two propositions.
\begin{proposition}
We have
\[
\delta^{k-1}(\mathcal{M}_{2-k}(\Gamma,R_S))\cap \widetilde{\mathcal{M}}_k(\Gamma,R_s)=\{0\}.
\]
\end{proposition}
\begin{proof}
Let $f\in \delta^{k-1}(\mathcal{M}_{2-k}(\Gamma,R_S))\cap \widetilde{\mathcal{M}}_k(\Gamma,R_s)$, so that in particular $f=\delta^{k-1}(g)$, for some $g\in \mathcal{M}_{2-k}(\Gamma,R_S)$. If $P\in X_\Gamma\setminus S_\Gamma$ were such that $\nu_P(g)<0$, then $\nu_P(f)<\frac{1-k}{h_P}$, which contradicts $f\in \widetilde{\mathcal{M}}_k(\Gamma,R_s)$. Therefore $\nu_P(g)\geq 0$ and a similar argument shows that $\nu_s(g)\geq 0$, for all $s\in S_\Gamma\setminus\{s_0\}$.

We now distinguish between the cases $k=2$ and $k\geq 3$. If $k=2$, then $f\in \widetilde{\mathcal{M}}_k(\Gamma,R_s)$ implies that $\nu_{s_0}(f)\geq 0$ which in turn yields that $\nu_{s_0}(g)\geq 0$. Therefore the meromorphic function $g$ has no poles on $X_\Gamma$, hence must be constant, and we conclude that $f=\delta(g)=0$. If $k\geq 3$ and $g\neq 0$, then Proposition \ref{proposition} now implies that
\[
\left\lfloor \nu_{s_0}(g)\right\rfloor=-1-\dim S_k(\Gamma)-\sum_{P \in X_\Gamma\setminus \{s_0\}}\left\lfloor \nu_{P}(g)\right\rfloor \leq -1-\dim S_k(\Gamma).
\]
In particular, $\nu_{s_0}(g)<0$, hence that $\nu_{s_0}(f)=\nu_{s_0}(g) \leq -1-\dim S_k(\Gamma)$, which contradicts $f\in \widetilde{\mathcal{M}}_k(\Gamma,R_s)$. Therefore we must have $g=0$, hence also $f=0$, ending the proof.
\end{proof}
\begin{proposition}
We have
\[
\delta^{k-1}(\mathcal{M}_{2-k}(\Gamma,R_S))+\widetilde{\mathcal{M}}_k(\Gamma,R_s)=\mathcal{M}_k(\Gamma,R_S).
\]
\end{proposition}
\begin{proof}
Let $f\in \mathcal{M}_k(\Gamma,R_S)$. If $f\in \widetilde{\mathcal{M}}_k(\Gamma,R_s)$, there is nothing to prove. Otherwise, one of the conditions (i)-(iii) in Theorem \ref{theorem} must be violated. We shall prove that it is always possible to add an element of $\delta^{k-1}(\mathcal{M}_{2-k}(\Gamma,R_S))$ to $f$ such that the result is contained in $\widetilde{\mathcal{M}}_k(\Gamma,R_s)$, which clearly implies the desired result.

Assume first that there exists $P\in X_\Gamma\setminus S_\Gamma$ such that $\nu_P(f)<\frac{1-k}{h_P}$ and choose $0\neq g \in \mathcal{M}_{2-k}(\Gamma,R_S)$. There exists $\varphi \in \mathcal{M}_0(\Gamma,R_S)$ such that $\nu_P(\varphi)=\nu_P(f)+\frac{k-1}{h_P}-\nu_P(g)$ (the right hand side is an integer by Lemma \ref{lemma}) and such that $\nu_Q(\varphi)>|\nu_Q(g)|$, for all $Q\in X_\Gamma \setminus \{P,s_0\}$. Then $\nu_P(\delta^{k-1}(\varphi\cdot g))=\nu_P(f)$, hence there exists $\alpha\in \mathbb C$ such that $\nu_P(f-\alpha\delta^{k-1}(\varphi\cdot g))>\nu_P(f)$ and $\nu_Q(f-\alpha\delta^{k-1}(\varphi\cdot g))\geq \nu_Q(f)$ for all $Q\in X_\Gamma\setminus \{s_0\}$. Repeating this step a finite number of times, we may thus ensure that, up to adding an element of $\delta^{k-1}(\mathcal{M}_{2-k}(\Gamma,R_S))$, we have $\nu_P(f)\geq \frac{1-k}{h_P}$, for all $P\in X_\Gamma\setminus S_\Gamma$. A similar argument shows that we may also assume that $\nu_{s}(f)\geq 0$ for all $s\in S_\Gamma\setminus \{s_0\}$.

Now assume that $\lfloor\nu_{s_0}(f)\rfloor<-\dim S_k(\Gamma)$ and choose $0\neq g\in \mathcal{M}_{2-k}(\Gamma,R_S)$. Up to possibly multiplying $g$ by a suitable modular function that only has a pole at $s_0$, we may assume that $\nu_P(g)\geq 0$, for all $P\in X_\Gamma\setminus \{s_0\}$. Moreover, by Proposition \ref{proposition} and since $X_\Gamma$ has genus zero , there exists $\varphi \in \mathcal{M}_0(\Gamma)$ such that $\nu_P(\varphi\cdot g)\geq 0$, for all $P\in X_\Gamma \setminus \{s_0\}$, and $\nu_{s_0}(\varphi\cdot g)=\nu_{s_0}(f)$. As before, one can now show that, up to adding an element of $\delta^{k-1}(\mathcal{M}_{2-k}(\Gamma,R_S))$, we have $\lfloor\nu_{s_0}(f)\rfloor \geq -\dim S_k(\Gamma)$, proving the proposition.
\end{proof}

\vspace{1mm}\noindent

\newpage


\section{The differential equations for the sunrise and banana integrals}
\label{app:sunban}

\subsection{The differential equations for the sunrise integrals}
\label{app:sun}
The matrices appearing in the differential equation in eq.~\eqref{eq:sun_GM} are
\beq\bsp
B^\sun(t) &\, = \frac{1}{6t(t-1)(t-9)}\begin{pmatrix}
3 (3+14t-t^2) & -9 \\\\
 (t+3) (t^3-15 t^2+75 t+3) & 3 (3+14t-t^2)
 \end{pmatrix}\,,\\
 D^\sun(t) &\, =\frac{1}{6t(t-1)(t-9)}\begin{pmatrix}
  6 (t-1) t & 0 \\
 (t+3) (t^3-9 t^2+63 t+9) & 3 (t-9) (t+1)
  \end{pmatrix}\,.
  \esp\eeq
  
 A basis of maximal cuts for the two-loop equal-mass sunrise integral in $d=2$ dimensions, i.e., a basis for the solution space of the differential operator in eq.~\eqref{eq:sunriseDO} is
\begin{equation}
	\begin{split}
  \label{eq:psi1_def}
  \Psi_1(t) & = \frac{4}{[(3-\sqrt{t})(1+\sqrt{t})^3]^{1/2}}\,{\EK}\left(\frac{t_{14}(t)t_{23}(t)}{t_{13}(t)t_{24}(t)}\right)\,,\\
  \Psi_2(t) & = \frac{4 i}{[(3-\sqrt{t})(1+\sqrt{t})^3]^{1/2}}\,{\EK}\left(\frac{t_{12}(t)t_{34}(t)}{t_{13}(t)t_{24}(t)}\right)\,,
  \end{split}
\end{equation} 
with $t_{ij}(t) = t_i(t)-t_j(t)$ and 
\begin{equation}\label{eq:SR_t_i_def}
  t_1(t) = -4\,,\quad t_2(t) = -(1+\sqrt{t})^2\,,\quad t_3(t) = -(1-\sqrt{t})^2\,, \quad t_4(t)=0\,,
\end{equation}
and $\EK(\lambda)$ denotes the complete elliptic integral of the first kind:
\beq
\EK(\lambda) = \int_0^1\frac{dt}{\sqrt{t(1-t)(1-\lambda t)}}\,.
\eeq

\subsection{The differential equations for the banana integrals}
\label{app:ban}

The matrices $B^\ban(x)$ and $D^\ban(x)$ entering the differential equation~\eqref{eq:banana_DEQ} are:
\begin{align}
B^\ban(x) &= 
\begin{pmatrix}
\frac{1}{x} & \frac{4}{x} & 0 \\
\frac{1}{4(1-x)} & \frac{1}{x}+\frac{2}{1-x} & \frac{3}{x}+\frac{3}{1-x} \\
-\frac{1}{8(1-x)} + \frac{1}{8(1-4x)} \,\,&\,\, -\frac{1}{1-x} + \frac{3}{2(1-4x)}
 \,\,&\,\, \frac{1}{x}+\frac{6}{1-4x}-\frac{3}{2(1-x)}
\end{pmatrix}\,,\\
D^\ban(x) &= \begin{pmatrix}
\frac{3}{x} & \frac{12}{x} & 0 \\
\frac{1}{1-x} & \frac{2}{x}+\frac{6}{1-x} & \frac{6}{x}+\frac{6}{1-x} \\
-\frac{1}{2(1-x)} + \frac{1}{2(1-4x)} \,\,&\,\, -\frac{3}{1-x}+\frac{9}{2(1-4x)} \,\,&\,\, \frac{1}{x}+\frac{12}{1-4x}-\frac{3}{1-x}
\end{pmatrix}\,.
\end{align}

A basis of the solution space for the differential operator $\cL_x^{\mathsf{ban},(3)}$ in eq.~\eqref{eq:L_ban_3} can then be chosen as (for $0\le t\le 1$) 
\begin{equation}
\begin{split}\label{eq:H1_to_Psi1}
I_1(x(t)) &\,=\frac{1}{3}\,t\,(\Psi_1(t)+\Psi_2(t))\,(\Psi_1(t)+3\Psi_2(t))\,,\\
J_1(x(t)) &\,=\frac{i}{3}\,t\,\Psi_1(t)\,(\Psi_1(t)+\Psi_2(t))\,,\\
H_1(x(t)) &\,= -\frac{1}{3}\,t\,\Psi_1(t)^2\,,
\end{split}
\end{equation}
where $\Psi_1(t)$ and $\Psi_1(t)$ are the maximal cuts of the sunrise integral, and $x(t)$ is defined in eq.~\eqref{eq:change_of_vars}.

\vspace{1mm}\noindent


\bibliographystyle{JHEP}
\bibliography{ellip}

\end{document}